\documentclass[11pt]{article}
\usepackage{latexsym}
\usepackage{amssymb,amsmath,amsthm}

\usepackage{authblk}

\usepackage{xcolor}
\usepackage{fullpage}
\definecolor{weborange}{rgb}{.8,.3,.3}
\definecolor{webblue}{rgb}{0,0,.8}
\definecolor{internallinkcolor}{rgb}{0,.5,0}
\definecolor{externallinkcolor}{rgb}{0,0,.5}
\usepackage[
hyperfootnotes=false,
colorlinks=true,
urlcolor=externallinkcolor,
linkcolor=internallinkcolor,
filecolor=externallinkcolor,
citecolor=internallinkcolor,
breaklinks=true,
pdfstartview=FitH,
pdfpagelayout=OneColumn]{hyperref}
\usepackage{cleveref}

\usepackage[
hyperfootnotes=false,
colorlinks=true,
urlcolor=externallinkcolor,
linkcolor=internallinkcolor,
filecolor=externallinkcolor,
citecolor=internallinkcolor,
breaklinks=true,
pdfstartview=FitH,
pdfpagelayout=OneColumn]{hyperref}

\usepackage[style=alphabetic,maxnames=6,url=false, isbn=false]{biblatex}
\newtheorem{theorem}{Theorem}

\newtheorem{proposition}[theorem]{Proposition}
\newtheorem{lemma}[theorem]{Lemma}
\newtheorem{claim}[theorem]{Claim}
\theoremstyle{definition}
\newtheorem{definition}[theorem]{Definition}
\newtheorem{algorithm}[theorem]{Algorithm}

\usepackage{tikz}
\newcommand*\circled[1]{\tikz[baseline=(char.base)]{
    \node[shape=circle,draw,inner sep=0.6pt] (char) {\tiny #1};}}

\def\circleS{\mathbin{\raisebox{1pt}{\circled {\rm S}}}}

\usepackage{accents}

\usepackage{complexity}

\def\vecone{{\vec 1}}
\def\veczero{{\vec 0}}
\def\Tpoly{{\mathrm{poly}}}
\def\Vzero{{\mathrm V ^ 0}}
\def\VL{{\mathrm{VL}}}
\def\VSL{{\mathrm{VSL}}}
\def\VNL{{\mathrm{VNL}}}
\def\VNC{{\mathrm{VNC}}}
\def\VSymKrom{{\mathrm{V\hbox{\rm-}SymKrom}}}

\def\tensor#1{{\mathbf {#1}}}

\def\IDelta{{\mathrm I \Delta}}
\def\Stheory{{\mathrm S}}
\def\Ttheory{{\mathrm T}}
\def\transpose{{\hbox{\rm\sf \scriptsize T}}}
\def\Neighbor{\hbox{\it Neighbor}}
\def\EdgeLabel{\hbox{\it EdgeLabel}}

\def\limplies{\rightarrow}
\def\liff{\leftrightarrow}
\def\pprime{{\prime\prime}}

\def\deltaGap{\delta}

\def\deltaPath{{\delta_{\mathrm{\scriptscriptstyle PATH}}}}
\def\deltaUCONN{{\delta_{\mathrm{\scriptscriptstyle UCONN}}}}
\def\Unique{\mathrm{Unique}}
\def\Symm{{\mathrm{Symm}}}
\def\MixRat{\mathrm{MixRat}}
\def\EdgeExp{\mathrm{EdgeExp}}
\def\SubsetConn{\mathrm{SubsetConn}}
\def\PathConn{\mathrm{PathConn}}

\definecolor{blue-violet}{rgb}{0.54, 0.17, 0.89}

\begin{document}

\title{\textbf{A Logspace-Constructive Proof of $\L=\SL$ }}
\author[1]{Sam~Buss}
\author[2]{Anant~Dhayal}
\author[3]{Valentine~Kabanets}
\author[4]{Antonina~Kolokolova}
\author[5]{Sasank~Mouli}
\affil[1]{\small Department of Mathematics, University of California, San Diego, La Jolla, CA, USA}
\affil[2]{Google, Kirkland, WA, USA}
\affil[3]{School of Computing Science, Simon Fraser University, Burnaby, B.C., Canada}
\affil[4]{Department of Computer Science, Memorial University of Newfoundland, St.\ John's, NL, Canada }
\affil[5]{Department of Computer Science, Indian Institute of Technology Indore, Madhya Pradesh, India}


\maketitle

\begin{abstract}
We formalize the proof of Reingold's Theorem that $\SL=\L$ \cite{Reingold:UndirectedConn_STOC} in the theory of bounded arithmetic $\VL$, which corresponds to ``logspace reasoning''. As a consequence, we get that $\VL=\VSL$, where $\VSL$ is the theory of bounded arithmetic for ``symmetric-logspace reasoning''. This resolves in the affirmative an old open question from Kolokolova~\cite{Kolokolova:thesis} (see also Cook-Nguyen~\cite{CookNguyen:book}). 

Our proof relies on the Rozenman-Vadhan alternative proof of Reingold's Theorem (\cite{RozenmanVadhan:DerandSquaring}). To formalize this proof in $\VL$, we need to avoid reasoning about eigenvalues and eigenvectors (common in both original proofs of $\SL=\L$). We achieve this by using some results from Buss-Kabanets-Kolokolova-Kouck\'y~\cite{BKKK:Expanders} that allow $\VL$ to reason about graph expansion in combinatorial terms.
\end{abstract}

\newpage
\tableofcontents

\newpage

\section{Introduction}\label{sec:Intro}

Symmetric logspace ($\SL$) has been studied for many years as a natural
complexity class, with the st-connectivity (reachability) problem (USTCON)
for \emph{undirected} graphs
as a natural complete problem for~$\SL$; that is, $\SL$~can be
defined as the class of problems that are many-one log-space
reducible to USTCON. Nisan and Ta-Shma~\cite{NisanTaShma:SLcomplement}
showed that $\SL$ is closed under complement. Finally, Reingold proved a breakthrough result that $\SL = \L$, i.e., that
symmetric logspace is equal to (deterministic)
logspace~\cite{Reingold:UndirectedConn_STOC,Reingold:UndirectedConn},
by showing that USTCON is in logspace.
(See \cite{Reingold:UndirectedConn} for a broader overview
of prior work on $\SL$ including expander graphs and derandomization.) Rozenman and Vadhan~\cite{RozenmanVadhan:DerandSquaring} later gave an alternative proof that USTCON is in logspace. Both proofs of $\SL=\L$ use expander graphs and linear algebra to analyze the expansion properties of various graph operations (usually by bounding certain eigenvalues).

Our main result is that $\SL=\L$ can be proved using ``logspace reasoning''. We show this by formalizing a variant of the Rozenman-Vadhan proof of $\SL=\L$ in the system of bounded arithmetic $\VL$, corresponding to logspace. This finally answers in the affirmative an old open question of \cite{Kolokolova:thesis,CookNguyen:book} of whether $\VL$ can prove $\L=\SL$.

\paragraph{Formalizing complexity results in logical theories.} First-order logical
theories, in particular fragments of bounded arithmetic, are known to prove a wide
range of complexity results. A number of non-trivial
algorithmic constructions, including Toda's theorem~\cite{BKZ:collapsingmodular}, 
the PCP theorem~\cite{Pich:PCP}, $\NC^1$-constructions of expander graphs~\cite{BKKK:Expanders}, 
the fact that $\NL=\coNL$~\cite{CookKolokolova:NL,CookNguyen:book},
constructions of hardness amplification~\cite{Jerabek:thesis}, 
and properties of the Arthur-Merlin
hierarchy~\cite{Jerabek:approximatecounting}, 
can be formalized and proved in various fragments of bounded arithmetic.
Some circuit lower bounds can be proved in bounded arithmetic theories,
including lower bounds on constant depth circuits~\cite{Krajicek:book,MullerPich:SuccinctLB}.
A recent unprovability result shows that a (weak) second order fragment of bounded
arithmetic cannot prove $\NEXP\subseteq\Ppoly$~\cite{ABM:NexpPpoly}.
Finally, recent work has shown that several complexity constructions can be ``reversed''
and are actually equivalent to some of the axioms used for bounded
arithmetic~\cite{CLO:reversemath}.  The present paper proves a
new formalization in bounded arithmetic, namely of $\L = \SL$.


\paragraph{Logspace reasoning in bounded arithmetic.}
What is the power of reasoning needed to prove that $\SL=\L$? Is ``logspace reasoning'' enough?  More formally, we ask for the weakest theory of bounded arithmetic that can formalize the proof of $\L = \SL$. It is natural to consider theories of bounded arithmetic that ``capture'' logspace reasoning.


There are several theories of bounded arithmetic for~$\L$; 
see the related work section below. Here we will use the theory~$\VL$ of 
Cook and Nguyen~\cite{CookNguyen:book}. 
This theory (like all other known theories for~$\L$)
can prove the existence of
second-order objects encoding logspace predicates, and can
prove the totality of logspace functions. Conversely, these
are the only predicates that can be proved to exist as second-order
objects, and the only provably total functions 
of~$\VL$, in the sense of~\cite{CookNguyen:book}.  In other words,
$\VL$~can reason using logspace properties and only logspace properties.
This means, in particular, that 
\begin{quote}
    the \emph{only} way that $\VL$ can
prove $\SL=\L$ is to give, at least implicitly, a constructively defined
logspace algorithm for $\SL$-predicates.
\end{quote}

An analogous situation is the case of $\VNL$, a theory corresponding to
nondeterministic logspace ($\NL$), that can prove $\NL = \coNL$~\cite{CookKolokolova:NL,CookNguyen:book}.
It does so by formalizing explicitly the Immerman-Szelepcs{\'e}nyi algorithm.

There are theories of bounded arithmetic that correspond
to~$\SL$ in the same way as $\VL$ corresponds to~$\L$.
The first one was defined by
Kolokolova~\cite{Kolokolova:thesis}, who
gave a second-order theory of bounded arithmetic,
$\VSymKrom$ that corresponds to~$\SL$. A natural question left open by~\cite{Kolokolova:thesis}  was whether $\VSymKrom = \VL$.
Cook and Nguyen reformulated $\VSymKrom$ by giving a theory $\VSL$ 
based on an axiom for undirected graph reachability 
(based on Zambella~\cite{Zambella:logspace}).
Their Open Problem IX.7.5~\cite{CookNguyen:book} asked whether
$\VL$ proves $\L = \SL$ and whether
$\VSL=\VL$.  The present paper answers these questions
affirmatively.


\paragraph{The difficulty of formalizing $\SL=\L$ in $\VL$.}
The difficulties arise from the fact that any $\VL$ proof can talk
only about logspace properties; it cannot talk about $\SL$ properties,
such as reachability in undirected graphs until \emph{after} $\SL=\L$ is
established. Any $\VL$ proof of $\SL=\L$ must be sufficiently 
constructive to use only concepts which are definable in logspace.
It cannot, for instance, use general exponential- or polynomial-time computations;
furthermore, we cannot hope to directly use concepts such as
determinants or eigenvalues which are only known to be in $\NC^2$ and
are conjecturally not logspace computable.

Fortunately, the arguments of Reingold~\cite{Reingold:UndirectedConn} and of Rozenman and Vadhan~\cite{RozenmanVadhan:DerandSquaring}, 
albeit somewhat complex, are
fairly straightforward and constructive. 
Reingold's proof used iteratively graph powering 
and a zig-zag/replacement
product
to transform an arbitrary connected graph into a good expander. His
methods increased the number of vertices in the graph polynomially,
but kept the graph constant degree.
Rozenman and Vadhan~\cite{RozenmanVadhan:DerandSquaring} gave
an alternate construction to prove that $\L = \SL$, using
derandomized squaring (denoted by~``$\circleS$'') to transform
any connected graph into a good expander; their
construction kept the number of vertices in the graph fixed but
increased the degree polynomially. As we shall explain below, from our point
of view, the Rozenman-Vadhan construction is easier to work with. 

The apparent difficulty is that both
proofs of $\SL=\L$ use expander graphs, and properties of expander graphs are
often proved using properties of eigenvalues and eigenvectors. On the other hand, Gaussian elimination and the determination of eigenvectors and eigenvalues are \emph{not} known to be computable in logspace, and thus  cannot be used by $\VL$~proofs.

Recently, Buss, Kabanets, Kolokolova, and Kouck\'y~\cite{BKKK:Expanders}
proved that the theory $\VNC^1$, which has logical complexity
corresponding to alternating logtime (uniform~$\NC^1$), can prove the existence of
expander graphs without needing to reason about eigenvalues and
eigenvectors.
Since $\VNC^1$ is a subtheory of~$\VL$, this provided
hope that $\VL$ can prove that $\L = \SL$. 

\paragraph{Our main result.}
The present paper succeeds in this direction and
gives a precise characterization of the logical complexity
of the proof that $\L=\SL$. 

\begin{theorem}[Main theorem, informal]
    $\VL$ proves $\L=\SL$. In addition, $\VSL=\VL$.
\end{theorem}

We show that a variation
of the Rozenman-Vadhan proof can be carried out in the second-order
bounded arithmetic theory~$\VL$. Carrying out this proof requires 
making the proof more explicitly constructive; namely,
the $\VL$-proof argues by contradictions proving the existence
of witnesses for 
existential statements. 
The $\VL$ proof (like the Rozenman-Vadhan proof)
works with mixing ratios
instead of second eigenvalues.  In addition,
the $\VL$ proof bypasses the use of the Cheeger inequality 
(which is not known to be provable in~$\VL$), and instead
uses a special case of the Cheeger inequality due
to Mihail~\cite{Mihail:Markov} that is known to be
provable in~$\VL$~\cite{BKKK:Expanders}. We provide more details next.

\paragraph{Our techniques.}
The proof by Rozenman and Vadhan~\cite{RozenmanVadhan:DerandSquaring} works mainly with the concept of the mixing ratio (see \Cref{def:mixingratio} below) that measures how fast a random walk on a given undirected regular graph~$G$ converges to the uniform distribution. Even though the mixing ratio happens to be equal to the second largest eigenvalue for the adjacency matrix of the graph $G$, this fact is \emph{not} used for most of the proof in \cite{RozenmanVadhan:DerandSquaring}. 

In particular, the analysis of the derandomized-squaring operation on graphs is done purely in terms of the mixing ratios of the input graphs, without using the equivalence between the mixing ratios and second largest eigenvalues. As we show, the basic linear algebra that is used in the analysis of derandomized squaring is simple enough to be proved in $\VL$. 

The only place where \cite{RozenmanVadhan:DerandSquaring} seems to rely on the use of eigenvalues in their proof is to show that every connected graph must have a nontrivial mixing ratio (at most $1-1/\poly(n,d)$, where $n$ is the number of vertices, and $d$ the degree of a given graph). Such a bound on the mixing ratio of a connected graph can be easily deduced by Cheeger's Inequality that relates the mixing ratio to the edge expansion of a graph. It is elementary to show (also in $\VL$) that any connected graph has nontrivial edge expansion. Thus, if $\VL$ could prove Cheeger's inequality, we could prove this step of the Rozenman-Vadhan argument in $\VL$ as well. 

Unfortunately, it is not known if $\VL$ can prove Cheeger's inequality in general. However, a version of Cheeger's inequality (for undirected regular graphs with enough self-loops around each vertex) can be proved in~$\VL$! This was done in \cite{BKKK:Expanders}, based on \cite{Mihail:Markov}; the approach in \cite{BKKK:Expanders} was to use the concept of \emph{edge expansion} to analyze graph operations (like replacement product, powering, and tensoring), and the results of \cite{Mihail:Markov} were used to relate edge expansion and mixing ratio, via a version of Cheeger's inequality. This version of the Cheeger-Mihail inequality can be used to get the required bound on the mixing ratio of any given regular graph, placing this part of the Rozenman-Vadhan proof in $\VL$.

But we are not done yet. The Rozenman-Vadhan argument also needs a sequence of explict expander graphs, of growing size and degree, to use in the successive applications of the derandomized-squaring operation. We need such expander graphs to be constructible in $\VL$. Luckily, we can again use the results of \cite{BKKK:Expanders} to argue that the required expander graphs can be constructed (and their expansion properties proved) in $\VL$.

Finally, we need to show in $\VL$ how to use the Rozenman-Vadhan transformation of a given input graph~$G$ into an expander~$\tilde{G}$ in order to solve the connectivity problem for $G$. Given a simple recursive construction of $\tilde{G}$ from~$G$ (obtained by repeatedly applying the derandomized squaring to the previous graph, using an appropriate expander as an auxiliary graph), we can show that $\tilde{G}$ is definable from~$G$ in~$\VL$. This allows us to complete the proof of $\SL=\L$ in~$\VL$.

We note that it may be possible to formalize Reingold's proof of $\SL=\L$ in $\VL$ as well, with some extra work. The original proof in \cite{Reingold:UndirectedConn_STOC} relied on the complicated analysis of the zig-zag product from the famous paper by Reingold, Vadhan, and Wigderson~\cite{wigderson2002entropy}, which used eigenvalues. However, based on \cite{RozenmanVadhan:DerandSquaring}, Reingold, Trevisan and Vadhan~\cite{DBLP:conf/stoc/ReingoldTV06} later gave a simpler analysis of the zig-zag product in terms of the mixing ratios (without using eigenvalues), which is sufficient for Reingold's proof of $\SL=\L$. This simpler analysis appears to be formalizable in~$\VL$ (as it is very similar to the Rozenman-Vadhan analysis of derandomized squaring, which we show in this paper to be formalizable in~$\VL$). The other ingredient in Reingold's proof, graph powering, can be easily analyzed in terms of mixing ratios as well. As above, we can use the Cheeger-Mihail inequality to show that any connected regular graph has a nontrivial mixing ratio, to start off Reingold's logspace transformation of a graph $G$ into an expander~$G'$. A somewhat tricky part in Reingold's proof is to argue that $G'$ can be used to answer the connectivity question for~$G$, by a \emph{logspace} algorithm. This part seems more complicated that the corresponding part in the Rozenman-Vadhan proof of $\SL=\L$, and may be challenging to implement in $\VL$, but probably could be. However, we do not pursue this approach in the present paper.     


\paragraph{Related work.}
Several early authors discussed the undirected graph reachability problem,
including \cite{JLL:newproblems,Schaefer:complexitySAT,AKLLR:RandomWalks}, but
the first systematic study of symmetric computation, including the definition
of~$\SL$, was due to
Lewis and Papadimitriou~\cite{LewisPapadimitriou:symmetriccomputation}
in 1982.
Savitch's theorem implies that $\SL$ is in $\SPACE(\log^2 n)$.
Aleliunas et~al.~\cite{AKLLR:RandomWalks}
and Borodin~\cite{BCDRT:InductiveCounting} gave \emph{randomized}
logspace algorithms for USTCON.
Nisan, Szemer\'edi, and Wigderson~\cite{NSW:USTCON} gave a
deterministic $\SPACE(\log^{1.5} n)$ algorithm;
Armoni et~al.~\cite{ATSWZ:USTCON} improved this to $\SPACE(\log^{4/3} n)$
and Trifonov~\cite{Trifonov:USTCON} even
more dramatically to $\SPACE(\log n \log\log n)$.

There have been several
bounded arithmetic theories proposed for logspace and symmetric logspace.
Clote and Takeuti~\cite{CloteTakeuti:NC} give the first theory for
logspace using a second-order theory~$\mathrm{S^{log}}$ of bounded arithmetic.
Zambella~\cite{Zambella:logspace} gave a
second, and more elegant, theory of bounded arithmetic,
based on a second-order theory of bounded arithmetic
including the $\Sigma^B_0$-rec axioms; the
$\Sigma^B_0$-rec axioms state
that there are second-order objects encoding polynomially
long paths in directed graphs of out-degree one.
Cook and Nguyen~\cite{CookNguyen:book} defined
an equivalent second-order theory~$\VL$ for logspace computation
with a reformulated version of the $\Sigma^B_0$-rec axioms.

Derandomized squaring has attracted renewed attention recently 
in the area of space-bounded derandomization; see, e.g.,  
\cite{MRSV21, DBLP:conf/focs/AhmadinejadKMPS20,DBLP:conf/focs/AhmadinejadPPSV23, DBLP:conf/focs/0001H0TW23, DBLP:conf/innovations/Cohen0MP25}. 

As mentioned, \cite{BKKK:Expanders}
proved that the theory~$\VNC^1$ can prove the existence of
expander graphs without using eigenvalues.
Viola~\cite{Viola:Impossible} also gives a construction of explicit expanders 
that avoids the use of eigenvalues based on the replacement product, along
an exposition of the Rosenman-Vadhan derandomized squaring proof for
Reingold's theorem.  Viola's motivation is to give
``simple'' explicit constructions rather than to formalize
results in bounded arithmetic.

\paragraph{Remainder of the paper.}
Section~\ref{sec:Prelims} introduces the needed preliminaries
for directed and undirected graphs, adjacency matrices,
vectors and tensors, edge expansion, and derandomized squaring.
Section~\ref{sec:LSLproof} gives the details of
the Rozenman-Vadhan proof of $\L=\SL$, both in its orginal
form and with a second version of the proof that will later be shown to
be formalizable in~$\VL$. 
Sections \ref{sec:Prelims} and~\ref{sec:LSLproof}
can be read independently of the rest of the paper for a self-contained
exposition of the proof of $\L=\SL$ that does not require
any knowledge of bounded arithmetic.  
Section~\ref{sec:FormalizeInVL}
starts with preliminaries for the bounded arithmetic theory~$\VL$,
and reviews the needed results of~\cite{BKKK:Expanders} about
edge expanders in~$\VNC^1$, including a ``Cheeger-Mihail lemma''
that can be formalized in $\VNC^1$ in lieu of the full Cheeger lemma.
It then sketches how to formalize the
proof of $\L=\SL$ in $\VL$. It follows that
$\VL$ and $\VSL$ are equal (\Cref{thm:VSLisVL}). Section~\ref{sec:SedrakyanCheeger}
proves that the Sedrakyan lemma and one direction of
the Cheeger Inequality
are provable in~$\VNC^1$ and thereby in~$\VL$.
The former is needed for the $\VL$ proof
described in Section~\ref{sec:FormalizeInVL}.

\section{Preliminaries}\label{sec:Prelims}

\subsection{Graphs and expansion}

A good source on expander graphs is \cite{hoory06}. 
This paper will consider both directed and undirected graphs. The statement
of $\L=\SL$ uses undirected graphs, and the expander graph constructions
of~\cite{BKKK:Expanders} are stated in terms of undirected graphs.
However, following Rozenman-Vadhan~\cite{RozenmanVadhan:DerandSquaring}, our proofs
depend on directed graphs.

An undirected graph is represented as
$G = (V,E)$ where $V$~is a set of vertices, and $E$~is a multiset of
\emph{edges} $\{u,v\}$ with $u,v \in V$.  It is allowed for $G$ to
have multiple (``parallel'') edges between $u$ and~$v$; i.e., $E$~may contain multiple copies
of $\{u,v\}$. Self-loops are allowed; namely, it is permitted
that $u=v$. The degree of a vertex $v\in G$ is the number of incident edges.
(A self-loop counts as a single incident edge.) A directed graph is \emph{$k$-regular}
if each vertex has degree~$k$.

A directed graph $G=(V,E)$ has
$E$ as a multiset of directed edges $\langle u,v \rangle$ where $u,v \in G$.
It is again permitted that a directed graph may contain self-loops
as well as multiple directed edges
from $u$ to~$v$.  A directed graph~$G$ is \emph{$k$-inregular},
respectively \emph{$k$-outregular}, if each vertex has indegree (respectively,
outdegree) equal to~$k$.  The graph is \emph{$k$-regular} if it is both
$k$-inregular and $k$-outregular.

An undirected graph $G=(V,E)$ can be converted into a directed
graph $G^\prime = (V, E^\prime)$ by replacing each
non-self-loop edge $\{u,v\}$ with the pair of directed edges $\langle u, v\rangle$
and $\langle v, u \rangle$ and replacing each self-loop $\{u,u\}$ with
a directed self-loop $\langle u,u \rangle$. If $G$ is $k$-regular, then
so is~$G^\prime$.

The proof that $\L = \SL$ uses directed regular graphs extensively. As we shall see in Section~\ref{sec:RVproof}, an arbitrary undirected graph~$Y$ can be converted into a regular undirected graph, and then further converted into
a regular directed graph~$X$ with similar connectivity properties. This lets constructions for regular directed graphs apply also to
arbitary undirected graphs. This enables proving results about $\SL$, namely about
reachability in undirected graphs.

We define edge expansion for both undirected and directed graphs $G=(V,E)$.
We say there is an edge from $u$ to~$v$ in~$G$ provided
that $E$ contains $\{ u, v \}$ if $G$ is undirected or that
$E$ contains $\langle u ,v \rangle$ if $G$ is directed.
For $U\subseteq V$, let $\overline U$~be $V\setminus U$. Then $E(U,\overline U)$
is the multiset of edges from $U$ to~$\overline U$, namely the
edges from a vertex $u\in U$ to a vertex $v \notin U$.

\begin{definition}[Edge expansion]\label{def:edgeExpansion}
Let $G$ be a $d$-regular graph on $n\ge 2$ vertices, either directed or undirected.
The \emph{edge expansion} of~$G$ is defined as:
\begin{equation}\label{eq:edgeExpansionDef}
\min_{\genfrac{}{}{0pt}{1}{\emptyset\ne U\subset V}{|U|\le n/2}}
     \frac{|E(U,\overline{U})|}{d \cdot |U|}
~=~
\min_{\emptyset\ne U\subsetneq V}\frac{|E(U,\overline{U})|}{d \cdot \min\{|U|,|\overline{U}|\}}
\end{equation}
\end{definition}

Note that if an undirected graph with edge expansion~$\epsilon$ is
converted to a directed graph, it still has edge expansion~$\epsilon$.
We often omit stating whether a graph is directed or undirected, when the
results hold in both cases.  A directed graph is \emph{connected} provided any two
vertices are connected by a directed path.

The developments in Sections \ref{sec:Intro} and~\ref{sec:RVproof}
follow the treatment in~\cite{RozenmanVadhan:DerandSquaring}
closely. We will take care to make the details of
the proofs sufficiently clear in order to later show how they
can be formalized in~$\VL$.

\begin{theorem}\label{thm:connected}
Let $G$ be a connected $d$-regular graph on $n$ vertices. Then the edge
expansion of~$G$ is at least $2/(dn)$.
\end{theorem}

\begin{proof}
This follows immediately from the definitions, since $E(U,\overline U)$ must
be nonempty.
\end{proof}

Let $G=(V,E)$ be a $d$-regular graph on $n=|V|$ many vertices.
When convenient, the vertices are identified with the integers
$[n] = \{1,\ldots, n\}$. The (normalized)
\emph{adjacency matrix} for~$G$ is the $n\times n$ matrix~$M$
with entries $M_{i,j}$ equal to the number of edges in~$E$ from
vertex~$j$ to vertex~$i$ divided by~$d$. All entries in~$M$ are
non-negative, and since $G$ is $d$-regular,
each row sum and column sum of~$G$ is equal to~$1$. The adjacency
matrix can be viewed as a transition matrix for a random walk in the
graph. Namely, if a $n$-vector~$\vec v$ represents a probability
distribution on the vertices~$V$, then $M \vec v$ gives the probability
distribution obtained after one step of a random walk starting
from the distribution~$\vec v$.

We use the term ``adjacency matrix'' for the just-defined
normalized version of the adjacency
matrix.  The \emph{unnormalized adjacency matrix} for a $d$-regular graph~$G$
is just $d\cdot M$: the $(i,j)$ entry in the unnormalized adjacency matrix
is the number of edges from $j$ to~$i$.

For $\vec v = \langle v_1, \ldots, v_n\rangle$ a vector,
$\|\vec v \|$ denotes the 2-norm $\bigl(\sum_i v_i^2\bigr)^{1/2}$.
We use $\langle \vec v, \vec w \rangle$ for the inner product
of two vectors, namely $\sum_i v_i w_i$.  We write
$\vec v \perp \vec w$ to denote that $\vec v$ and~$\vec w$
are orthogonal, i.e., that $\langle \vec v, \vec w \rangle = 0$.
We call $\vec v$ a \emph{(probability) distribution} if $\sum_i v_i = 1$
and each $v_i\ge 0$.
The vector~$\vecone$ is the $n$-vector
$\langle 1/n, \ldots 1/n\rangle$ corresponding
to the uniform distribution on the vertices~$V$.
Note that $\|\vecone\| = 1/{\sqrt n}$. The zero vector
is denoted~$\veczero$.

The next theorem follows immediately from the
fact that the adjacency matrix has non-negative
entries and all of $M$'s row and column sums
equal~1.

\begin{theorem}\label{thm:adjacencyTrivial}
Let $M$ be the adjacency matrix of a regular graph.
\begin{description}
\setlength{\itemsep}{0pt}
\setlength{\topsep}{0pt}
\item[\rm (a)] If $\vec v$ is a probability distribution,
then so is~$M \vec v$.
\item[\rm (b)] $M \vecone = \vecone$.
\item[\rm (c)] If $\vec v \perp \vecone$, then also
$M\vec v \perp \vecone$.
\end{description}
\end{theorem}

\begin{lemma}[Sedrakyan's Lemma]\label{lem:Sedrakyan}
Suppose $u_i, v_i \in \mathbb R$ and $v_i > 0$ for all $1\le i \le n$.
Then
\[
\sum_i \frac {u_i^2}{v_i} ~\ge~ \frac {\bigl(\sum_i u_i\bigr)^2}{\sum_i v_i}.
\]
\end{lemma}

\begin{proof}
This follows from the Cauchy--Schwarz inequality
$\langle \vec u', \vec v' \rangle^2 \le \|\vec u'\|^2\cdot \|\vec v'\|^2$
for the vectors $\vec u'$ and~$\vec v'$ with entries
$u'_i = u_i/\sqrt{v_i}$ and $v'_i = \sqrt{v_i}$.
\end{proof}

The \emph{norm} of a matrix~$M$, denoted $\| M\|$, is the least $\alpha$ such that
$\|M \vec v\| \le \alpha \|\vec v\|$ for all~$\vec v$.

\begin{theorem}\label{thm:normM}
Let $M$ be the adjacency matrix of a regular graph.  Then
$\|M\| = 1$.
\end{theorem}

\begin{proof}
For any $\vec v$, we have:
\begin{align*}
\|M \vec v\|^2 ~&=~ \sum\nolimits_i\,\bigl( \sum\nolimits_j M_{i,j} v_j \bigr) ^2 \\
  &\le~ \sum\nolimits_i\,
  \sum\nolimits_{j \colon M_{i,j}\neq 0}
      \frac{(M_{i,j} v_j)^2}{M_{i,j}}
        \tag{By Sedrakyan's lemma  since $\sum_j M_{i,j} = 1$} \\
  &=~ \sum\nolimits_i\,\sum\nolimits_j M_{i,j} v_j^2 \\
  &=~ \sum\nolimits_j\,\bigl( \sum\nolimits_i M_{i,j} \bigr) v_j^2 ~=~ \sum\nolimits_j v_j^2 ~=~ \|\vec v\|^2.
  \qedhere
\end{align*}
\end{proof}

\begin{definition}[Mixing ratio]\label{def:mixingratio}
Let $G$ be a $k$-regular graph (directed or undirected) with
adjacency matrix~$M$. Let $\lambda \ge 0$.  The graph~$G$
has \emph{mixing ratio}~$\lambda$ provided that
\[
\| M \vec v\| ~\le \lambda \cdot \|\vec v\|
\]
holds for all $\vec v \perp \vecone$.  The minimum such~$\lambda$
is called \emph{the mixing ratio} of~$G$ and is denoted~$\lambda(G)$.
\end{definition}
\noindent
When we say ``$G$~has mixing ratio~$\lambda$'', we mean that the
mixing ratio of~$G$ is $\le \lambda$.

By Theorem~\ref{thm:normM}, the mixing ratio~$\lambda$ is $\le 1$.
It is easy to see that the mixing ratio of an
undirected graph~$G$ is equal to
the second largest, in absolute value, eigenvalue of its adjacency matrix $M$. However,
Definition~\ref{def:mixingratio}
defines the mixing ratio without referring to eigenvalues or
eigenvectors. This has the advantage that the bounded arithmetic
theory~$\VL$ can work with the concept of mixing ratio, even though
$\VL$ is not known to be able to prove properties of
eigenvalues and eigenvectors.

\begin{definition}
A graph $G$ is a $(n,d,\lambda)$-graph if
it has $n$~vertices, is $d$-regular,
and has mixing ratio (at most)~$\lambda$.
\end{definition}

We let $J_n$ denote the $n\times n$ matrix with
all entries equal to~$1/n$. Note that $J_n$
is the adjacency matrix for the complete graph with all self-loops on
$n$~vertices; equivalently, $J_n$~is the transition matrix
for a random walk on $n$ vertices.

\begin{theorem}\label{thm:JCdecomp}
Let $M$ be the adjacency matrix of an $(n,d,\lambda)$-graph
with $0<\penalty10000 \lambda\le\penalty10000 1$.
Then
\begin{itemize}
\item[\rm (a)] {\normalfont \cite{RozenmanVadhan:DerandSquaring}} $M$ is equal to $(1-\lambda) J_n + \lambda C$ where
$\|C\|\le 1$. 
\item[\rm (b)] $M$ is equal to $J_n + \lambda D$ where $\|D\| \le 1$ and
all of the row and column sums
of~$D$ are equal to zero.
\item[\rm (c)] $\|M\| = \|J_n\| = \|J_n+D\| = 1$.
\end{itemize}
\end{theorem}

\begin{proof}
We let $C = (M - (1-\lambda)J_n)/\lambda$. By Theorem~\ref{thm:adjacencyTrivial}, 
$M\vecone = J_n \vecone = \vecone$. Hence $C \vecone = \vecone$.
Now suppose $\vec v \perp \vecone$, so $J_n \vec v = \veczero$
and therefore $C \vec v = M \vec v / \lambda$.
By hypothesis, $\|M \vec v\| \le \lambda \|\vec v \|$ and by Theorem~\ref{thm:adjacencyTrivial},
$M \vec v \perp \vecone$. It follows that
$C \vec v \perp \vecone$ and $\| C \vec v\| \le \|\vec v\|$.
That proves~(a).

Now let $D = (M - J_n)/\lambda = C-J_n$. From the argument for part~(a),
we have that $D \vecone = \veczero$ and that for all
$\vec v \perp \vecone$, we have $D \vec v \perp \vecone$ and
$\|D \vec v\| \le \|\vec v\|$.  The row and column sums of $M$ and~$J_n$ are equal to~1; hence the row
and column sums of~$D$ equal zero. Parts (b) and~(c) follow.
\end{proof}

\begin{theorem}\label{thm:connectedBis}
Let $G$ be a connected, $d$-regular, undirected graph on $n$ vertices.
The mixing ratio~$\lambda$ of~$G$ is
at most $1 - 2/(dn)^2$.
\end{theorem}

\begin{proof}
By Theorem~\ref{thm:connected}, the edge expansion of~$G$ is
at least $2/(d n)$.  The well-known Cheeger inequality, which
is stated below as Theorem~\ref{thm:Cheeger},
immediately implies that $1-\lambda \geq 2/(dn)^2$.
\end{proof}

\begin{theorem}\label{thm:connectedTri}
Let $G$ be a connected, $d$-regular, directed graph on $n$ vertices
with a self-loop at each vertex.
The mixing ratio~$\lambda$ of~$G$ is
at most $1 - 1/(d^4n^2)$.
\end{theorem}

\begin{proof}
Let $M$ be the adjacency matrix for the directed graph~$G$. Form the undirected
graph~$H$ that has adjacency matrix $M^\transpose M$ by
letting the edges of~$H$ be pairs $\{x,y \}$ such that
there is a vertex~$u$ such that $\langle x,u \rangle$ and $\langle y,u \rangle$ are
(directed) edges in~$M$. We allow multiedges and self-loops in~$H$; the
multiplicity of an edge $\{ x,y\}$ in~$H$ depends on the number
of pairs of edges $\langle x, u\rangle$ and $\langle y,u \rangle$ in~$G$.
In particular, $H$~is $d^2$-regular.
Since $G$ has self-loops, every edge $\langle x,y \rangle$ in~$G$ gives rise
to the edge $\{x,y\}$ in~$H$. Therefore, $H$~is connected.

Suppose $\|\vec v \| =1$ and $\vec v\perp \vecone$.  We wish
to show $\| M \vec v \| \le 1-1/(d^4 n^2)$.
By Theorem~\ref{thm:connectedBis}, $H$~has mixing ratio $\le 1-2/(d^4 n^2)$.
Thus, $\|M^\transpose M \vec v\| \le  1-2/(d^4 n^2)$.
Therefore,
\[
\| M \vec v\|^2 ~=~ \langle M\vec v, M \vec v \rangle 
~=~  \langle \vec v, M^\transpose M \vec v \rangle 
~\le~ \|M^\transpose M \vec v\|
~\le~ 1-2/(d^4 n^2)
~\le~ (1-1/(d^4 n^2))^2,
\]
where the first inequality holds by Cauchy--Schwarz since
$\| v \| = 1$.
Thus $\| M \vec v \| \le 1-1/(d^4 n^2)$ as desired.
\end{proof}

Rozenman and Vadhan~\cite{RozenmanVadhan:DerandSquaring} prove slightly stronger
versions of Theorems \ref{thm:connectedBis} and~\ref{thm:connectedTri}
with $1-2/(dn)^2$ and $1 - 1/(d^4n^2)$
replaced with $1-1/(d n^2)$ and $1-1/(2 d^2 n^2)$.

Unfortunately, it is open whether the Cheeger inequality and thus
Theorems \ref{thm:connectedBis} and~\ref{thm:connectedTri}
can be proved in the theory~$\VL$. However, there is a weaker form
of the Cheeger lemma with a constructive proof by Mihail~\cite{Mihail:Markov} that
can be used instead, and this Cheeger-Mihail theorem is
known to be provable in~$\VNC^1$~\cite{BKKK:Expanders}.  The Cheeger-Mihail
theorem states that the Cheeger inequality holds for an undirected regular
graph of even degree~$d$ provided that each vertex has at least
$d/2$ self-loops. The proof of Theorem~\ref{thm:connectedTri} can thus
be formalized in~$\VNC^1$ (and hence in~$\VL$) provided that the
graph~$H$ has sufficiently many self-loops at each vertex.  For more
details, see Theorems \ref{thm:connectedBisVNC1} and~\ref{thm:connectedTriVNC1},
in Section~\ref{sec:FormalizeInVL} on formalizations in~$\VL$.

\subsection{Derandomized Squaring}

We now review the definition of derandomized squaring,
$X \circleS G$.

\begin{definition}[Proper labeling]
Let $G$ be a $d$-regular graph. Suppose each edge in~$G$ has been assigned a label from a set of cardinality~$d$;
this is called the \emph{label set} of~$G$.

If $G$ is undirected, this is called a \emph{proper labeling}
provided that, for each vertex of~$G$, the $d$ incident edges have
distinct labels.

If $G$ is directed, this is called a \emph{proper labeling} provided that,
for each vertex, the $d$ incoming edges have distinct labels, and the
$d$ outgoing edges have distinct labels.\footnote{Our ``proper labeling''
is the same as what Rozenman and Vadhan~\cite{RozenmanVadhan:DerandSquaring}
call a ``consistent labeling''.}
\end{definition}

Note that a proper labeling of an undirected $d$-regular graph
immediately gives a proper labeling of the corresponding
directed $d$-regular graph.

\begin{definition}
Suppose $G=(V,E)$ is a properly labeled
$d$-regular (directed) graph.
For $v$ a vertex in~$V$ and $i$ an edge label, we write
$v[i]$ to denote the vertex~$w$ that is reached by taking a single step in~$G$
from vertex~$v$ along the outgoing edge labeled~$i$.
\end{definition}

Now suppose that $X$ and $G$ are properly labeled,
regular, directed graphs and that
$X$~has $N$~vertices and is $K$-regular and that
$G$~has $K$~vertices and is $D$-regular.
Without loss of generality, the vertex set of~$X$ is~$[N]$, the vertex set of~$G$ is~$[K]$, the
label set of~$X$ is also~$[K]$, and the label set of~$G$ is~$[D]$.

The squared graph $X^2$ is the graph on the same set of vertices~$[N]$ and degree~$K^2$. It has an edge between vertices $v$ and~$u$ iff $u$ can be reached from~$v$ in two steps, i.e., if there exist $1\leq i,j\leq K$ such that $u=w[j]$ for $w=v[i]$. In contrast, the derandomized square of~$X$, using an auxiliary graph $G$, contains an edge between vertices $v$ and~$u$ only if $u$ can be reached from~$v$ in two steps $1\leq i,j\leq K$, where the second step~$j$ is computed from the first step~$i$ using~$G$. Namely, vertex~$j$ must be a neighbor of vertex~$i$ in~$G$. (If $G$ were a complete graph on~$[K]$, we would end up with the usual squared graph~$X^2$. We get fewer edges, however, if $G$ has degree $D<K$.) More formally, we have the following definition.

\begin{definition}[Derandomized squaring]
Let $X$ and~$G$ be as above.
The \emph{derandomized-squaring product} of $X$ and~$G$
is denoted $X \circleS G$; it is a $KD$-regular directed graph
on $N$~vertices (it has the same set of vertices as~$X$).
The edge labels of $X \circleS G$ will be written as
pairs $(i, j)$ where $i\in [K]$ and $j\in [D]$; in other words,
$i$ and~$j$ are edge labels for~$X$ and ~$G$ respectively.
We also view $i$ as a vertex in~$G$.

The edges of $X \circleS G$ are defined as follows. Let $v\in [N]$
be a vertex of $X \circleS G$ (and a vertex of~$X$) and let $(i,j)$
be an edge label for~$X \circleS G$. Working
in~$X$, let $w = v[i]$, namely the vertex of~$X$ reached by following
edge~$i$ from~$v$.  Then, working in~$G$, let $k = i[j]$, namely the vertex
in~$G$ obtained by following edge~$j$ from~$i$. Then, working in~$X$ again,
let $u = w[k]$, namely the vertex reached by following edge~$k$ from~$w$.
This $u$ is a vertex in~$X \circleS G$; by definition, $X \circleS G$ has
an edge from $v$ to~$u$ with edge label $(i,j)$. As a shorthand notation,
we can write $u = (v[i])[i[j]]$.
\end{definition}

\begin{theorem}[\cite{RozenmanVadhan:DerandSquaring}]\label{thm:circleSWellDefn}
Suppose $X$ and~$G$ are properly labeled, regular, directed graphs
as above. Then $X \circleS G$ is a properly labeled,
regular, directed graph.
\end{theorem}

\begin{proof}
The proof of Theorem~\ref{thm:circleSWellDefn} is completely elementary.
By construction, $X \circleS G$ is $KD$-outregular and the
outgoing edges of any vertex in $X \circleS G$ have distinct labels.
Conversely, consider whether a vertex~$u$ in $X \circleS G$ has
an incoming edge with the label $(i,j)$. Since $X$ is properly
labeled, there are unique vertices $w$ and~$u$ of~$X$ such that
$u = w[i[j]]$, and $w = v[i]$. Therefore, there is exactly
one incoming edge to~$u$ in $X \circleS G$ with label $(i,j)$.
It follows that $X\circleS G$ is $KD$-inregular and properly labeled.
\end{proof}

\subsection{Tensors}

Let $U$ and $V$ be (real) vector spaces, of finite dimensions $m$ and $n$, respectively.
Let $\vec u = \langle u_1,\ldots, u_m\rangle \in U $ and
$\vec v = \langle v_1,\ldots, v_n\rangle \in V$.
We write $\vec u \otimes \vec v$ to denote the rank one tensor product of
$\vec u$ and~$\vec v$; this
is by definition equal to the $m \times n$ outer product matrix
$ \vec u (\vec v^T)$. Here $\vec v^T$ is the
transpose of~$\vec v$: vectors are column vectors,
so $\vec v^T$ is a row vector.

More generally,
any $m \times n$ matrix represents a tensor~$\tensor w$.
(We use boldface to represent tensors.) For $i\in[m]$ and $j\in[n]$,
we let $\tensor w_{i,j}$ denote the $(i,j)$-entry of~${\mathbf w}$.
Alternately, a tensor can
viewed as a $mn$-vector; however, usually the matrix representation is
more useful. In either representation, the tensors over $U$ and~$V$ form
an $mn$-dimensional vector space.

A linear map on tensors can be represented by an $mn\times mn$ matrix~$M$.
The entries of~$M$ are denoted $M_{(i,j),(i^\prime, j^\prime)}$ for
$i,i^\prime \in [m]$ and $j,j^\prime \in [n]$. Then $M \tensor w$ is the tensor
with entries
$(M \tensor w)_{i,j} = \sum_{i^\prime,j^\prime} M_{(i,j),(i^\prime, j^\prime)} \tensor w_{i^\prime,j^\prime}$.
The magnitude~$\|\tensor w \|$ of~$\tensor w$ is the magnitude of the
$mn$-vector representation of~$\tensor w$.
It is easy to check that $\| \vec u \otimes \vec v \|$ is
equal to $\|\vec u\| \cdot \| \vec v\|$.
The norm~$\|M\|$ of a linear map~$M$ on
tensors is the same as the matrix 2-norm of~$M$ viewed as an
$mn \times mn$ matrix. Thus $\|M\|$ equal to the least $\nu$ such that
$\| M \tensor w \| \le \nu \| \tensor w\|$ for all~$\tensor w$.

Suppose now that $N$ and~$N^\prime$
are $m\times m$ and $n \times n$ matrices, respectively.
Their tensor product $N \otimes N^\prime$ is the tensor map with matrix
representation given by~$M$ with
$M_{(i,j),(i^\prime,j^\prime)} = N_{i,i^\prime} N^\prime_{j,j^\prime}$.
For rank one tensors $N$ and~$N^\prime$, straightforward computation shows that
the tensor product $N \otimes N^\prime$ acts on $\vec u \otimes \vec v$
by sending it to $N \vec u \otimes N^\prime \vec v$.
In addition, $\|N \otimes N^\prime\|$ is equal to $\|N\| \cdot \| N^\prime\|$.

We also work with mappings between vectors and tensors.
The first is the projection operation~$P$ that projects away
the second component of tensors.  More precisely, $P \tensor w$
is the $m$-vector $\vec v$ such that $v_i = \sum_j \tensor w_{i,j}$.
The operator~$P$ can be expressed as the $m \times mn$-matrix with entries
$P_{i^\prime,(i,j)} = \delta_{i,i^\prime}$ where $\delta_{i,i^\prime}$ equals~1
if $i=i^\prime$ and equals~0 otherwise.  (Of course, $P$~depends on $m$ and~$n$,
but we suppress this in the notation.) It is easy to check
that $\|P\| = \sqrt n$.

Conversely, the lifting operation~$L$ is defined by 
$L \vec u:=\vec u \otimes\penalty10000 \langle 1/n,\ldots,1/n \rangle$.
The lifting operation~$L$ can be presented by the $mn \times m$-matrix
with all entries equal to $L_{(i,j),i'}=\delta_{i,i'}/n$.
We have $\|L\| = 1/{\sqrt n}$, and in fact
$\|L \vec u \| = \|u\|/{\sqrt n}$ for all~$u$. The composition
$P \circ L$ is the identity transformation, and of course
has magnitude $\|P\circ L\| = 1$.


\section{Expansion for \texorpdfstring{$\L=\SL$}{L=SL}}\label{sec:LSLproof}

This section details the Rozenman and Vadhan proof of
$\L = \SL$. 

\subsection{One application of derandomized squaring}
The core construction of Rozenman and Vadhan's
proof that $\L = \SL$ required showing bounds on the
mixing ratio of derandomized squaring. Their bound is
re-proved in the next theorem.

\begin{theorem}[\cite{RozenmanVadhan:DerandSquaring}]\label{thm:RVexpand}

Suppose $X$ is a properly labeled, $K$-regular, directed
graph with vertices~$[N]$.
Let $G$ be a properly labeled, directed $(K, D, \mu)$-graph.
Thus $X\circleS G$ is a properly labeled, directed, $KD$-regular
graph with vertices~$[N]$.
Let $0 < \lambda < 1$.
\begin{itemize}
\item[\rm (a)] If $X$ has mixing ratio~$\lambda$,
then $X\circleS G$ has mixing ratio at most $f(\lambda,\mu)$ where
\begin{equation}\label{eq:FunctionF}
f(\lambda, \mu) ~=~ \mu + (1-\mu) \lambda^2.
\end{equation}
\item[\rm (b)] In particular, let $A$ and~$M$ be the adjacency
matrices of $X$ and~$X\circleS G$, respectively, and suppose
there is an $N$-vector $\vec v$ such that
$\vec v \perp \vecone$ and
$\| M \vec v\| > f(\lambda,\mu) \| \vec v \|$.
Then there is an $N$-vector $\vec u$ such that $\vec u \perp \vecone$
and $\| A \vec u \| > \lambda \|\vec u\|$.
In fact, this will hold for at least one of
$\vec u = \vec v$ or $\vec u = A \vec v$.
\end{itemize}
\end{theorem}

Part~(a) can be more succinctly stated as
saying that if $X$~is an $(N,K,\lambda)$-graph
and $G$~is a $(K,D,\mu)$-graph, then $X\circleS G$ is
an $(N,KD,f(\lambda,\mu))$-graph. The theorem is stated
in a roundabout way, however, so we can better discuss
its provability in the theory~$\VL$.

\begin{proof}
Part~(a) follows from~(b), so we prove~(b). The adjacency matrix~$M$
can be viewed as computing a single random step in~$X \circleS G$.
We view $M$ as acting on
an $N$-vector~$\vec v$ representing
a probability distribution on vertices~$[N]$.
The $\ell$-th entry $(\vec v)_\ell$ of $\vec v$ gives
the probability of being at vertex~$\ell$.
Following the proof of~\cite{RozenmanVadhan:DerandSquaring},
we shall describe $M$ as the transition matrix for a
five-step random process applied to a vertex $v \in [N]$ where
$v$~is chosen at random according to~$\vec v$. Without loss of generality, $X$ has label set~$[K]$,
and $G$ has vertex set~$[K]$ and label set~$[D]$.

The five-step random process starts at a vertex~$v$ of~$X$
and ends at a vertex~$w$ of~$X$; note that $v,w \in [N]$.
As is described below, the $s$-th intermediate step (for $s=1,2,3,4$)
of the random process transitions to a state
of the form $\langle u, i\rangle$, where $u\in [N]$ and $i\in[K]$.
We interpret $u$~as a vertex of~$X$ and $i$~as a vertex of~$G$.

Suppose the input~$v$ to the random
process is chosen according to
a fixed probability distribution~$\vec v$ on~$[N]$.
For $s=1,2,3,4$, let $\tensor w_s$ be the tensor that
gives the probability distribution on the
result $\langle u,i \rangle$ of the $s$-th step of
the random process;
that is, $(\tensor w_s)_{u,i}$ is the
probability of being in the state $\langle u, i \rangle$ after
the $s$-th step. Of course, $\tensor w_s = M_s \vec v$ for
some $NK\times N$ matrix~$M_s$.

\begin{itemize}
\item[\rm 1.] Choose a random vertex $v\in [N]$ of~$G$ according
to the probability distribution~$\vec v$.
Choose $i \in [K]$ uniformly at random
and go to the state represented by the pair $ \langle v, i\rangle$.
This corresponds to picking a random outgoing edge~$i$ of~$v$.
Note $i$ is also a vertex of~$G$.

The transition matrix for this
step is the $NK\times N$ matrix~$L$; namely, this step is represented
by the lifting operation $\vec v \mapsto L \vec v$, so
$\tensor w_1$ equals~$L \vec v$.

\item[\rm 2.] Transition deterministically to the
state represented by $\langle v[i], i \rangle$.
This corresponds to
following the edge labeled~$i$ in~$X$.

The transition matrix for this step is the
$NK\times NK$-matrix~$\tilde A$ defined with
$\tilde A_{(u,i),(u^\prime,i)} = 1$ if there is an edge
in~$X$ from $u$ to~$u^\prime$ labeled~$i$ and with
all other entries of $\tilde A$ equal to zero.
Since $X$ is $K$-regular and properly labeled,
$\tilde A$ is a permutation matrix. (This corresponds to
the fact that this is a deterministic, reversible step.)
In fact, when $\tilde A$ is viewed as acting on a $N\times K$ matrix,
multiplication by~$\tilde A$ permutes entries
within each column of the $N\times K$ matrix.
Then $\tensor w_2 = \tilde A \tensor w_1$.
Since $i$ was chosen uniformly at random,
$\tensor w_2$ has the form $\vec z \otimes \vecone$ for some $\vec z$.

\item[\rm 3.] Choose a random neighbor~$k$ of~$i$ in~$G$, and transition to
the state given by $\langle v[i], k\rangle$.
In other words, choose a random $j\in [D]$
and transition to $\langle v[i], i[j]\rangle$.

The transition matrix for this step is equal to $\tilde B := I_N \otimes B$,
where $I_N$~is the $N\times N$ identity matrix
and $B$~is the $K\times K$ transition matrix for~$G$.
We have $\tensor w_3 = \tilde B \tensor w_2$.

\item[\rm 4.] Deterministically transition to the pair $\langle v[i][k], k\rangle$.
This corresponds to following the edge in~$X$ labeled~$k$ from the vertex $v[i]$
to reach the vertex $v[i][k]$.

The transition matrix for this step is the same
matrix~$\tilde A$ that was used for step~2.
We have $\tensor w_4 = \tilde A \tensor w_3$.

\item[\rm 5.] Deterministically go to state $v[i][k]$ (discarding
the second component of the state).

The transition matrix for this step is the $N\times NK$ projection operator~$P$.
We let $\vec w$ be the $N$-vector
$\vec w = P \tensor w_4$; this makes $\vec w$ the
probability distribution on the vertices of~$[N]$ at the end
of the five step random process.
\end{itemize}
Putting this all together shows that the transition matrix~$M$
for $X \circleS G$ is equal to
\begin{equation}\label{eq:RVexpandPf}
M ~=~ P \tilde A (I_N\otimes B) \tilde A L .
\end{equation}
By Theorem~\ref{thm:JCdecomp}, $B$~is equal to
$(1-\mu) J_K + \mu C$ where $\|C\|\le 1$. Thus,
\[
M ~=~ (1-\mu) P \tilde A (I_N\otimes J_K)\tilde A L
      + \mu P \tilde A (I_N\otimes C)\tilde A L .
\]
We have $I_N\otimes J_K$ is equal to $LP$.
And $P\tilde A L$ is equal to~$A$, the transition matrix for~$X$.
Therefore,
\begin{eqnarray}
\nonumber
M &=& (1-\mu) P \tilde A LP \tilde A L
      + \mu P \tilde A (I_N\otimes C)\tilde A L  \\
\nonumber
  &=& (1-\mu) (P \tilde A L)^2 + \mu P \tilde A (I_N\otimes C) \tilde A L \\
\label{eq:Mexpand}
  &=& (1-\mu) A^2 + \mu P \tilde A (I_N\otimes C)\tilde A L.
\end{eqnarray}
Recall that $\|P\| = \sqrt K$, and
$\|L\| = 1/\sqrt K$. Also, $\|I_N \otimes C\| \le 1$ since $\|C\| \le 1$.
Finally, $\|\tilde A\| = 1$ since $\tilde A$ is a permutation
matrix. It follows that
$\|P \tilde A (I_N\otimes\penalty10000 C)\tilde A L\| \le 1$.

We can now prove part~(b) of the theorem.
Suppose $M$ has mixing ratio $>\penalty10000 f(\mu,\lambda)$, so there is
an $N$-vector~$\vec z$ such that $\vec z \perp \vecone$ and
\[
\|M\vec z\| > (\mu + (1-\mu)\lambda^2)\|\vec z\| .
\]
Therefore, by~(\ref{eq:Mexpand}) and
since $\|P \tilde A (I_N\otimes C)\tilde A L\| \le 1$,
\begin{equation}\label{eq:RVexpandPf2}
   (1-\mu) \|A^2 \vec z\| + \mu \|\vec v\| > (\mu+(1-\mu)\lambda^2) \|\vec v \|,
\end{equation}
whence $\|A^2 \vec z \| > \lambda^2 \|\vec z\|$.  Thus
$\|A \vec u\| > \lambda \|\vec u\|$ for $\vec u$~equal to (at least) one
of $\vec z$ and $A\vec z$. In either case, $\vec u \perp \vecone$, either since
$\vec z \perp \vecone$ or, if $\vec u=A\vec z$, by
Theorem~\ref{thm:adjacencyTrivial}(c).
\end{proof}

The proof that $\L=\SL$ will iterate
the construction of Theorem~\ref{thm:RVexpand}. Namely,
it starts with a graph~$X$ with adjacency matrix~$M_X$
that has mixing ratio at
most $\lambda(X) = 1 - 2/(dn)$ and applies derandomized squaring
$k = O(\log n)$ times to form a graph
\[
Z ~=~ (( \cdots ( X \circledS G_1 ) \circledS G_2 ) \cdots ) \circledS G_k.
\]
It is then argued that if the $G_i$'s have appropriate mixing ratios,
then the adjacency matrix~$M_Z$ of~$Z$
has small mixing ratio,
and hence $Z$ has good expansion. 
Specifically, it is proved
if there is an $N$-vector $v \perp \vecone$ so that
$\|M_Z \vec v\| > \lambda(Z) \|\vec v\|$,
then there is an $N$-vector $\vec u\perp \vecone$
so that $\| M_X \vec u \| > \lambda(X) \|\vec u\|$.

The proof method of Theorem~\ref{thm:RVexpand}
is sufficiently constructive that it will allow
the theory~$\VL$ to prove expansion properties for iterated
applications of derandomized squaring. This is shown
in Section~\ref{sec:FormalizeInVL}. First, however, we give the rest
of the proof that $\L = \SL$.

\subsection{Iterating derandomized squaring and proving \texorpdfstring{$\L = \SL$}{L=SL}}\label{sec:RVproof}
Rozenman-Vadhan's proof that $\L=\SL$ proceeds in three stages.
The first (easy) stage starts with an undirected graph~$Y$, not
necessarily regular, and produces a
directed graph~$X$ which is 4-regular, has a loop at each vertex
and is properly labeled, such that the reachability problem for~$Y$ is
easily reducible to the reachability problem for~$X$. The construction replaces
each degree~$d$ node~$v$ in~$Y$ with $d$-many nodes denoted $\langle v,i \rangle$
for $0\le i <d$.  The four outgoing edges for the node $\langle v,i\rangle$ in~$X$
and their labels are:
\begin{itemize}
\setlength{\itemsep}{0pt}
\setlength{\parsep}{0pt}
\item From $\langle v, i \rangle$ to $\langle v, i{-}1 \bmod d \rangle$ with label~1.
\item From $\langle v, i \rangle$ to $\langle v, i{+}1 \bmod d \rangle$ with label~2.
\item From $\langle v, i \rangle$ to $\langle w, j \rangle$ with label~3,
where $w$~is the $i$-th neighbor
of~$v$ and $v$~is the~$j$-th neighbor of~$w$.
\item From $\langle v, i \rangle$ to itself (a self-loop) with label~4.
\end{itemize}
Let $N$ be the number of vertices in~$X$.

The second stage starts with the 4-regular directed graph~$X$
and repeatedly augments it by derandomized squaring with graphs~$G_i$.
The $G_i$'s are defined in terms of graphs~$H_i$ which
have the following properties, for some constant $Q=4^q$:
\begin{itemize}
\item $H_i$ is a properly labeled
  $(Q^i, Q, 1/100)$-graph.
\item The edge relation in~$H_i$ is computable uniformly in
simultaneous logspace and polylogarithmic time in the size of~$H_i$.
That is, the
function $\langle v,x\rangle \mapsto\penalty10000 v[x]$, for $v$~a vertex and $x$~an edge label
is computable in simultaneous space $O(i)$ and time $\Tpoly(i)$. Note
that $O(i)$ many bits are needed to specify $v$ and~$x$.
\end{itemize}

It is well-known that such graphs~$H_i$ exist for some suitably large (constant)
value for~$Q$. 
In fact, it follows from
Theorem~32 of~\cite{BKKK:Expanders} that
$\VNC^1$ proves the existence of such graphs.

Set 
\[
m_0 = \lceil 100 \log N \rceil.
\]
For $i\le m_0$, let $G_i = H_i$.
For $i>m_0$, let 
\[
G_i = (H_{m_0 + 2^{i-m_0} -1})^{2^{i-m_0}}.
\]
Then define $X_1 = X^q$ so that $X_1$ is $Q$-regular (since $X$ is 4-regular).
For $i\ge 1$, let
\begin{equation}\label{eq:XiDefn}
X_{i+1} ~=~ X_i \circleS G_i.
\end{equation}
Figure~\ref{fig:XGdegrees} shows the
sizes and degrees of the graphs $X_i$ and~$G_i$
up to 
\[
m_1 = m_0 + \log\log N + O(1).  
\]
For $i\le m_0$,
$G_i$~has constant degree~$Q$.
For $i>m_0$, $G_{m_0+i}$~has
degree $Q^{2^i}$.  It takes $2^i \cdot \lceil \log Q\rceil$
bits to specify an
edge of~$Q_{m_0+i}$.  The edge relation
for~$G_{m_0 + i}$ is computable in space
$O(2^i + \log m)$: namely, it is computed by
traversing $2^i$ many edges of~$H_{m_0 + 2^i-1}$.
The edge relation of~$H_{m_0 + 2^{i-m_0} -1}$
is computable in space~$O(m_0 + 2^i)$ and
only extra $i$~many bits are needed to keep track
of the number of edges traversed.

\begin{figure}
\centering
\begin{tabular}{c|c|c|c|c|c}
     & Number of & Degree & Number of & Degree & Upper bound  \\
$i$ & vertices of $X_i$ & of $X_i$ & vertices of $G_i$ & of $G_i$ & on $\lambda(G_i)$ \\
\hline
1 & $N$ & $Q = 4^q$ & $Q$  & $Q$ & $1/100$ \\
2 & $N$ & $Q^2$ & $Q^2$  & $Q$ & $1/100$  \\
3 & $N$ & $Q^3$ & $Q^3$  & $Q$ & $1/100$  \\
$\vdots$ & $\vdots$ & $\vdots$ & $\vdots$ & $\vdots$  & \vdots \\
$m_0$ & $N$ & $Q^{m_0}$ & $Q^{m_0}$  & $Q$ & $1/100$  \\
$m_0+1$ & $N$ & $Q^{m_0+1}$ & $Q^{m_0+1}$  & $Q^2$ & $(1/100)^2$  \\
$m_0+2$ & $N$ & $Q^{m_0+3}$ & $Q^{m_0+3}$  & $Q^4$ & $(1/100)^4$  \\
$m_0+3$ & $N$ & $Q^{m_0+7}$ & $Q^{m_0+7}$  & $Q^8$ & $(1/100)^8$  \\
$\vdots$ & $\vdots$ & $\vdots$ & $\vdots$ & $\vdots$ & $\vdots$  \\
$m_0+\ell$ & $N$ & $Q^{m_0+2^\ell-1}$ & $Q^{m_0+2^\ell-1}$
       & $Q^{2^\ell}$ & $(1/100)^{2^\ell}$
\end{tabular}
\caption{The number of vertices and
the $k$-regularity degree for
the graphs $X_i$ and~$G_i$, and the mixing ratio~$\lambda(G_i)$.
The process stops once
$m_1 = m_0+ \penalty1000 \ell$ with $\ell = \log\log N +\penalty 10000 O(1)$.}
\label{fig:XGdegrees}
\end{figure}

Suppose $Y$, and hence~$X$, is connected. (Or if not, restrict attention to
a connected subset of~$X$.)
We claim
that $\lambda(X_{m_0}) < 3/4$ and that
$\lambda(X_{m_1})<1/N^2$, where $m_1 = m_0 + \ell$ with
$\ell = \log \log N + O(1)$. These claims will be proved
by induction using:
\begin{proposition}\label{prop:BoundsOnF}
Let $f(\lambda, \mu)$ be as defined in~{\rm (\ref{eq:FunctionF})}. Then,
for $\mu, \lambda \in (0,1)$,
\begin{enumerate}
\setlength{\itemsep}{0pt}
\setlength{\itemsep}{0pt}
\item[\rm a.] $1-f(1-\gamma,\mu) \ge (3/2) \gamma$
if $\gamma \le 1/4$ and $\mu \le 1/100$.
\item[\rm b.] $f(\lambda,\mu) < \lambda^2 + \mu$.
\item[\rm c.] $f(\lambda,\mu)$ is monotonically increasing
in both $\lambda$ and~$\mu$.
\end{enumerate}
\end{proposition}
The ``spectral gap'' of $X_i$ is defined to equal $1-\lambda(X_i)$.
Part~a.\ of the proposition will be used to prove that the
spectral gaps of the $X_i$'s are increasing by a factor of~$3/2$
for small values of~$i$, namely for $i<m_0$.
Part~b.\ will be used for $X_i$'s with $i>m_0$.
\begin{proof}
For part~a., from $f(\lambda, \mu) = 1-(1-\lambda^2)(1-\mu)$,
we have
\[
1-f(1-\gamma,\mu) ~=~ (1-(1-\gamma)^2)(1-\mu)
  ~=~ \gamma (2 -\gamma) (1-\mu),
\]
and
\[
\hbox{$(2 -\gamma) (1-\mu) > 3/2$ ~if~ $\gamma \le 1/4$ and $\mu \le 1/100$.}
\]

Parts b.\ and c.\ are immediate from the definition of $f(\lambda,\mu)$.
\end{proof}
Note that the bounds in Proposition~\ref{prop:BoundsOnF} are not tight.

\begin{claim}\label{claim:lambdaX}
Assume $X$ is connected. Let $N\ge 2$ and $\ell = 10 + \lceil\log\log N\rceil$ and $m_1 = m_0 + \ell$.
Then
$\lambda(X_{m_0}) < 3/4$ and $\lambda(X_{m_1}) < 1/N^2$.
\end{claim}

\begin{proof}
To prove the first part of claim, we have
$\lambda(X_1) \le 1 - 1/(256 N^2)$ by Theorem~\ref{thm:connectedTri}
using the fact that $X$ has degree~4 and has a self-loop at every vertex.
Thus by Proposition~\ref{prop:BoundsOnF},
\[
1 - \lambda(X_i) \ge \min\{ 1/4, 1{-}\penalty10000 (3/2)^i/(256 N^2) \}.
\]
Since $m_0 = \lceil 100 \log N \rceil$, the first inequality follows. (The
constant 100 is not optimal.)

The second part of the claim
is proved  
by letting $\lambda_i = (64/65)\cdot\penalty10000 (7/8)^{2^i}$ and showing
that $\lambda(X_{m_0+i}) \le \lambda_i$ by induction on~$i$.
The base case of $i=0$ follows from the first inequality since
$(64/65)(7/8) > 3/4$.
Let $\mu_i = (1/100)^{2^i}$.
For the induction step, we have 
\[
X_{m_0+i+1} = X_{m_0+i} \circleS G_{m_0+i}
\]
where $\lambda(G_{m_0+i}) \le (1/100)^{2^i}$.
Thus
\begin{align*}
\lambda( X_{m_0+i+1}) & \le  \lambda_i^2 + (1/100)^{2^i} \qquad\qquad&&\hbox{by the induction hypothesis}\\
~& \le~  \lambda_i^2 \cdot \frac {65}{64}
     &&\hbox{since $\displaystyle (1/100)^{2^i} \le \lambda_i^2/64$ } \\
& \le~  \Bigl(\frac{64}{65}\Bigr)^2 \cdot
           \Bigl( \frac78 \Bigr)^{2^{i+1}} \cdot \frac{65}{64} 
            &&\hbox{by definition of $\lambda_i$}\\
& =~  \lambda_{i+1}.
\end{align*}
This proves the induction step.  To finish off the claim,
it suffices to show $1/N^2 > (7/8)^{2^\ell}$. Taking
logarithms and by $\ell = 10 + \lceil\log \log N\rceil$, it suffices to show
\[    2 \log N ~<~ \log (8/7) \cdot 2^{10} \cdot 2^{\log \log N}
    ~=~ 2^{10}\cdot\log(8/7)\log N .
\]
Simple computation shows this holds.
\end{proof}


The directed graph $X_{m_1}$ has such nice expansion properties that it lets us solve
the reachability problem for~$X$ (hence for~$Y$) essentially trivially.  Namely,
vertex~$y$ is reachable in~$X$ from vertex~$x$ if and only there is edge
from $x$ to~$y$ in~$X_{m_1}$.
Since $X_{m_1}$ has polynomial degree, all outgoing edges from~$x$
can be checked in logarithmic space. Therefore, to finish the proof that
$\L = \SL$ we need to establish two facts: first,
Theorem~\ref{thm:expandImpliesConnected} shows that,
since $\lambda(X_{m_1})$ is small enough, to determine whether $y$ is
reachable from~$x$ in~$X$, it is sufficient to check for an edge from
$x$ to~$y$ in~$X_{m_1}$. Second, Lemma~\ref{lem:LogSpaceAlgForEdges} shows that
there is a log space algorithm to determine the $i$-th edge outgoing from~$x$.

\begin{theorem}\label{thm:expandImpliesConnected}
Suppose $G$ is a directed graph with $N$ vertices and $\lambda(G) <\penalty10000 1/N^{2}$.
Then, for every pair of vertices $u$ and~$v$ in~$G$, there is an edge from $u$ to~$v$.
\end{theorem}

\begin{proof}
Let $M$ be the adjacency matrix of~$G$. Let $u$ be a vertex of~$G$
and let $e_u$ be the vector with a~1 in the $u$~position and all other entries~0.
That is, $e_u$~is the probability distribution assigning probability~1 to
being at vertex~$u$.  We must show that all entries
of the vector $M {e_u}$ are non-zero. In fact, we will show
that all entries of~$M e_u$ are $\ge 1/N -1/N^2$.  Let $\vec v$ be the vector
with all entries equal to~$1/N$. Note that $M \vec v = \vec v$. It suffices
to show that all the entries in the vector $Me_u -\vec v = M(e_u -\vec v)$ have absolute
value $\le 1/N^2$.

The squared magnitude $\|e_u-\vec v\|^2$ is equal to
$(N-1)\cdot (1/N^2) + (1-1/N)^2 = (N-1)/N < 1$. Since
$\lambda(G) \le 1/N^2$, this means 
\[
\|M(e_u-\vec v)\|^2\le 1/N^4.
\]
Let $s$ be the entry in $M(e_u-\vec v)$ with
maximum absolute value. We have $s^2 \le \|M(e_u-\penalty10000 \vec v)\|^2$,
whence $s^2 < 1/N^4$. Thus $|s| < 1/N^2$. This proves the theorem.
\end{proof}

Claim~\ref{claim:lambdaX} and Theorem~\ref{thm:expandImpliesConnected}
imply that, in order to check whether $x$ and~$y$ are in the same
connected component of~$X$, it suffices to check whether there
is an edge from $x$ to~$y$ in~$X_{m_1}$. We now describe how to
traverse edges in~$X_{m_1}$.

The directed graph~$X_{m_1}$ was defined from~$X_1$ by repeatedly
applying derandomized squaring with the graphs~$G_i$, see~(\ref{eq:XiDefn}).
From the definition of derandomized squaring,
$X_m$~has in-degree and out-degree as given in Figure~\ref{fig:XGdegrees}.
From a starting vertex in~$X_m$, the
outgoing edges are labeled by $m$-tuples
\begin{equation}\label{eq:XmEdge}
w ~=~ ( z, a_1, a_2, \ldots, a_{m-1} ),
\end{equation}
where $z$~is an label for an edge of~$X_1$
and each $a_i$~is an label for an edge of~$G_i$.
The edge label~$w$ can be viewed as being obtained
by appending $a_{m-1}$ to~$w'$, where $w' = ( z, a_1, a_2, \ldots, a_{m-2} )$
is an label for an edge in~$X_{m-1}$ and $a_{m-1}$~is
an label for an edge in~$G_{m-1}$.
The value~$z$ and the values of~$a_i$ for $i \le m_0$ are $< Q = 4^q$
and are specified with $2q$~bits each.
For $i = m_0+j$,
$a_j <\penalty10000 Q^{2^j} =\penalty10000 4^{q 2^j}$ and is specified with
$ q 2^{j+1}$ many bits.  Since $m_0 = O(\log N)$ and
$m \le m_1 = m_0 + \log \log N + O(1)$, the edge~$w$ is specified
with $O(\log N)$ bits, in keeping with the fact that
$X_{m_1}$ has polynomial degree.

Suppose $x$ is a vertex in~$X$ and $w$~is an edge label in~$X_m$
as in~(\ref{eq:XmEdge}). We wish to find the vertex~$y$ in~$X$ such
that the $w$-th outgoing edge from~$x$ in~$X_m$ goes to
vertex~$y$.  From the definition of derandomized squaring,
this is done by:
\begin{enumerate}
    \item Letting $w' = ( z, a_1, a_2, \ldots, a_{m-2} )$
and following the $w'$-th outgoing edge from~$x$ in~$X_{m-1}$
to reach a vertex~$u$ in~$X$,
\item Viewing $w'$ as a vertex in~$G_{m-1}$
and following the $a_{m-1}$-th outgoing edge from~$w'$
in~$G_{m-1}$ to reach a vertex~$w''$ in~$G_{m-1}$, and
\item Viewing $w''$ as an edge label in~$X_{m-1}$
and following the $w''$-th outgoing edge from $u$
to~$y$ in~$X_{m-1}$.  This is the desired vertex $y$
reached by following the $w$-th outgoing edge in~$X_m$.
\end{enumerate}

To find such a vertex $y$ in logspace, note that an edge $w = ( z, a_1,\ldots, a_{m-1} )$ in $X_m$ corresponds to a path of length $2^{m-1}$ in $X_1$. The recursive definition of~$X_m$ allows one to compute, in logspace, the edge label in position $b$, for $0\leq b < 2^{m-1}$, of this path in~$X_1$, for every given~$b$. We use $b=0$ for the first edge
on the path.

Indeed, the required path of length $2^{m-1}$ consists of its left subpath and its right subpath of length $2^{m-2}$ each, where the left subpath corresponds to taking the edge $w' = ( z, a_1, a_2, \ldots, a_{m-2} )$ in $X_{m-1}$, and the right subpath corresponds to taking the edge $w'[a_{m-1}]$ (computed using $G_{m-1}$) also in $X_{m-1}$. Thus, if $b < 2^{m-2}$ we recurse on the left subpath (with the same value~$b$), and otherwise, we recurse on the right subpath (with $b$ updated to $b-2^{m-2}$). Crucially, as noted in \cite{RozenmanVadhan:DerandSquaring}, the edge label in position $b$ in $X_1$ depends only on the edge label $w$ in $X_m$, and not on the initial vertex of the path in $X_1$. This allows us to get the following algorithm.

\begin{algorithm}\label{alg:XmEdge-b}
Algorithm $\EdgeLabel(b,w)$
\begin{tabbing}
\underline{Inputs:} 
\=$b$ expressed in binary notation as $(b_{m-2},\ldots, b_0)$, and \\
 \>$w = ( z, a_1,\ldots, a_{m-1} )$ an edge label for~$X_m$. \\
\underline{Output:} the $(b{+}1)$st edge label in the $2^{m-1}$-length path in $X_1$ corresponding to the edge $w$ in $X_m$. \\
\underline{Procedure:} \\
~~~~\= Set $y_{m-1} = ( z, a_1,\ldots, a_{m-1} )$ \\
\>For $i=m-2$ to $0$\\
\>~~~\=Split the $(i{+}2)$-tuple $y_{i+1}$ into the prefix $(i{+}1)$-tuple $y_i$, and the remaining component $a'$\\
\>\>If $b_i=1$ then set $y_i=y_i[a']$ (the $a'$-th neighbor of $y_i$ in graph $G_{i+1}$) \\
\>\>Endif\\
\> Endfor\\
\>Output $y_0$
\end{tabbing}
\end{algorithm}

Then, to compute the neighbor $y$ of $x$ in $X_1$, we go over all $b$ from $0$ to $2^{m-1}-1$, computing the $b$th edge label $A_b$ in the path in $X_1$, updating $x$ to $x[A_b]$.  We have the following pseudocode algorithm.

\begin{algorithm}\label{alg:XmEdge}
Algorithm $\Neighbor(x,w)$
\begin{tabbing}
\underline{Input:} \=$x$ a vertex in $X$, and \\
 \>$w = \langle z, a_1,\ldots, a_{m-1}\rangle$ an edge label for~$X_m$. \\
\underline{Output:} $y = x[w]$ in $X_m$, the vertex reached via the $w$-th outgoing edge in $X_m$. \\
\underline{Procedure:} \\
~~~\=Set $y = x$ \\
\>For $b=0$ to $2^{m-1}-1$\\
\>~~~\=Set $a=\EdgeLabel(b,w)$\\
\>\>Set $y=y[a]$ (the $a$-th neighbor of $y$ in graph $X_1$)\\
\>Endfor\\
\>Output $y$
\end{tabbing}
\end{algorithm}

This establishes:

\begin{lemma}[\cite{RozenmanVadhan:DerandSquaring}]\label{lem:LogSpaceAlgForEdges}
There is a logspace algorithm for traversing an edge in~$X_m$.
\end{lemma}

\noindent
The above constructions immediately give:

\begin{theorem}[\cite{Reingold:UndirectedConn_STOC,Reingold:UndirectedConn}]
$\L = \SL$.
\end{theorem}

\section{Formalization in \texorpdfstring{$\VL$}{VL}}\label{sec:FormalizeInVL}

\subsection{Preliminaries for bounded arithmetic}
For space reasons, we give only a high-level descriptions of the theories $\VNC^1$,
$\VL$ and $\VSL$. For more information on bounded arithmetic, the reader
should consult \cite{Buss:bookBA} and
Kraj\'\i\v cek~\cite{Krajicek:book} for a broad introduction,
and Cook-Nguyen~\cite{CookNguyen:book} for the definitions
and full development of $\Vzero$, $\VNC^1$, $\VL$ and $\VSL$.

There is a long history of formalizing results about low-level computational
complexity in bounded arithmetic. In fact, the original definitions of
$\IDelta_0$, $\Stheory^i_2$ and $\Ttheory^i_2$ were all motivated by the goal
of designing arithmetic theories corresponding to reasoning with
low-complexity concepts~\cite{Parikh:feasibility,Buss:bookBA}. 
As discussed in the introduction,
a variety of complexity upper and lower bounds have
been established within bounded arithmetic theories, mostly in
subtheories of~$\Ttheory_2$ (which is a theory that can reason
about concepts expressible in the polynomial hierarchy).
The most important such result for the present paper is that
the existence of expander graphs is provable in
$\VNC^1$~\cite{BKKK:Expanders}. The present paper
works primarily in~$\VL$, but also extensively in its subtheory $\VNC^1$.

The theories $\VNC^1$, $\VL$, and $\VSL$ are extensions of~$\Vzero$.
The theory~$\Vzero$ has logical strength that corresponds
to the complexity class~$\AC^0$.
It is a second-order (i.e., a two-sorted) theory of arithmetic; the two sorts are
the (nonnegative natural)
numbers (first-order objects) and the strings (second-order objects).
Strings can be viewed either as members of $\{0,1\}^*$ or as finite sets
of numbers.
The notation $X(i)$, where $X$ is a
string and $i\geq 0 $ is a natural number,
means the Boolean value of the $i^{th}$ entry in string~$X$.
Sometimes ``$i{\in}X$'' is written instead of ``$X(i)$''.
The constants $0$ and $1$ are number terms, and addition and multiplication
are number functions.
The string length function~$|X|$ maps strings
to numbers, and is defined to equal
one plus
the value of the largest element in~$X$ when viewed as a set.
The addition and multiplication functions map pairs of
numbers to numbers.
The equality relation is defined both for numbers and strings.
The axioms of~$\Vzero$ consist of a finite set
of ``BASIC'' open axioms defining simple properties of the
constants, relation symbols, and function symbols,
plus the extensionality axiom
and the $\Sigma^B_0$-Comprehension axioms
\[
\Sigma^B_0{\hbox{\rm -COMP:}} ~~~~
\exists X{\le}y\, \forall z{<}y \, (X(z)\liff \varphi(z))
\]
for any formula $\varphi$ in~$\Sigma^B_0$ not containing $X$ as
a free variable, but possibly containing free variables
other than~$z$.  
The notation $(\exists X{\le}y)\psi$
means the same as
$\exists X (|X|{\le}y \wedge \psi)$.
Quantifiers such as $\exists X{\le}t$,
$\forall z{<}t$ and $\forall z{<}t$ are
called \emph{bounded quantifiers}.
A $\Sigma^B_0$-formula is one
in which all quantifiers are bounded and
which contains no
second-order quantifier. Second-order
variables are allowed in a $\Sigma^B_0$-formula, but may not be quantified.
The $\Sigma^B_1$-formulas can be defined as the smallest class
of formulas containing $\Sigma^B_0$ and closed under application of
existential second-order (bounded) quantification. For instance, if $\varphi(X,Y)$
is a $\Sigma^B_0$-formula, then $(\exists Y{\le} t) \varphi(X,Y)$ is a
$\Sigma^B_1$-formula.\footnote{Since all first-order quantifiers
are bounded, it makes no essential difference for the following
discussion whether the second-order quantifiers in $\Sigma^B_1$-formulas
are required to be bounded since the existentially quantified second-order objects are applied
only to bounded first-order values.}

For connecting bounded arithmetic to computational complexity,
numbers are viewed as ranging over small values, namely the lengths
of strings.  Strings are viewed as inputs to algorithms.
It is well-known that the $\Sigma^B_0$-formulas $\varphi(X)$,
with only~$X$ free,
can express exactly properties that are in $\AC^0$: for this, the
input~$X$ is viewed as a binary string of length~$|X|$. Likewise, the $\Sigma^B_1$-formulas
can express exactly NP properties.

The logical strengths of $\VNC^1$, $\VL$ and $\VSL$ can be characterized as
follows.
Let $T$ denote one of the theories $\VNC^1$, $\VL$ or $\VSL$,
and let $\mathcal C$ denote the corresponding complexity class,
uniform~$\NC^1$ (alternating log time), log space $\L$,
or symmetric logspace $\SL$.  A formula $\varphi(X)$ is $\Delta^B_1$-definable with respect to $T$ provided that $T$ proves both $\varphi(X)$ and $\lnot \varphi(X)$
are equivalent to $\Sigma^B_1$-formulas.  Then the $\Delta^B_1$-formulas
express precisely the predicates in the complexity class~$\mathcal C$.
A string-valued function $F(X)$ is $\Sigma^B_1$-defined by~$T$
provided there is a $\Sigma^B_1$-formula $\varphi(X,Y)$ which expresses
the condition $Y=F(X)$ such that $T$~proves $\forall X\, \exists Y\, \varphi(X,Y)$.
Then, the $\Sigma^B_1$-definable functions of~$T$ are precisely the
functions whose bit-graphs are in the complexity class~$\mathcal C$.  In particular, the $\Sigma^b_1$-definable
functions of $\VNC^1$ and $\VL$ are the alternating logtime computable functions and the
logspace computable functions (respectively). The theory~$T$
can be conservatively extended to include predicate
symbols for its $\Delta^B_1$-definable predicates,
and to include function symbols for its $\Sigma^b_1$-definable functions. The theory~$T$
proves comprehension also for formulas~$\varphi(z)$ which
are $\Sigma^B_0$ in the expanded language with
these additional predicate and function symbols.
Indeed, $T$~proves comprehension for formulas
which are $\Delta^B_1$ (even $\Delta^B_1$ in the
expanded language).

The theories $\VNC^1$, $\VL$ and $\VSL$ are axiomatized in a way
that ensures these just-stated properties hold. The definition of
$\VNC^1$ can be found in
Cook-Morioka~\cite{CookMorioka:NCone} or Cook-Nguyen~\cite{CookNguyen:book}.
The axioms for $\VL$ and $\VSL$ assert the existence of paths
in graphs.  A path can be viewed as a sequence of numbers: a sequence
is encoded by a string~$Z$ with $Z(i,y)$ intended to denote that $y$ is
the $i$-th element of a sequence. Since the $i$-th element is unique,
the least~$y$ such that $Z(i,y)$ is used. The notation $(Z)^i$ is
used for this: specifically,
\[
y=(Z)^i ~\Leftrightarrow~
     (Z(i,y) \lor y = |Z|) \land (\forall y'{<}y)\lnot Z(i,y).
\]
Here $Z(i,y)$ means $Z(\langle i,y \rangle)$ for a suitable
logspace computable pairing function
$\langle \cdot, \cdot \rangle$ on numbers.

A string~$X$ represents the edge relation of a directed graph on
the vertices $[a+1] = \{0,1,\ldots,a\}$, by letting $X(y,z)$
denote that there is an edge from $y$ to~$z$, where $y,z \le a$.

A path of length~$b$ in the graph starting at vertex~0
is represented by a string~$Z$ so that
\[
\deltaPath(a, b, X, Z) ~\leftrightarrow~
(\forall i{\le} b) (\, (Z)^i \le a \,)
 \land (\forall i{<} b) X(\, (Z)^i, (Z)^{i+1}\,) \land (Z)^0 = 0.
\]
Also, let $\Unique(a,X)$ be $(\forall x{\le}a)(\exists ! y{\le}a)X(x,y)$
expressing that all vertices have a unique outgoing edge.  Then,
$\VL$ is axiomatized as $\Vzero$ plus the axiom
\[
\Unique(a, X) ~\limplies~ (\exists Z)\deltaPath(a,a,X,Z).
\]
This asserts the existence of a path of length~$a$ in the
directed graph with edge relation~$X$, starting from the
vertex~0.\footnote{We follow
\cite{CookNguyen:book} in including the condition
$\Unique$, but all that is really needed is that every vertex
has degree at least one. There is also no loss of axiomatic strength
in requiring the first two parameters of $\deltaPath$ to both equal~$a$. }

The theory $\VSL$ is axiomatized in terms of connectivity
and paths in \emph{undirected} graphs. Accordingly, define
\[
\Symm(a,X) ~\Leftrightarrow~(\forall x{\le} a)X(x,x) \land
    (\forall x{\le}a)(\forall y {\le} a)(X(x,y)\liff X(y,x))
\]
to express that the graph with edge relation~$X$ is
symmetric and has self-loops. A connectivity predicate $C(\cdot)$
is a string that identifies vertices connected to the vertex~$0$.
It is paired with a set of paths from~$0$ to each vertex connected to~$0$.
Namely, letting $Z$ now take triples as input,
define $Z^{[u]}$ to be the binary predicate
so that $Z^{[u]}(i,y)\Leftrightarrow Z(u,i,y)$.
In other words, the string $Z^{[u]}$ is the $u$-th slice of the string~$Z$.
Undirected connectivity is thus witnessed by:
\begin{eqnarray} \label{eq:deltauconn} 
\lefteqn{\deltaUCONN(a, X, C, Z) ~\Leftrightarrow~  
       C(0) \land (\forall x {\le}a)(\forall y {\le}a)(C(x) \land X(x,y) \limplies C(y))} \notag \\
& &
    \qquad\qquad\qquad\qquad\qquad \land ~
    (\forall x{\le}a)[C(x) \limplies \deltaPath(a,a, X, Z^{[x]})
       \land (Z^{[x]})^a = x ].
\end{eqnarray}
This expresses that the connectivity predicate is closed under the edge
relation, and that every vertex~$x$ connected to~$0$ has a path from~$0$.
The theory~$\VSL$ is axiomatized by~$\Vzero$ plus the axiom
\begin{equation}\label{eq:VSLaxiom}
\Symm(a,X) ~\limplies~ (\exists C)(\exists Z)\deltaUCONN(a,X,C,Z) .
\end{equation}


\subsection{Formalization in \texorpdfstring{$\VNC^1$}{VNC1} and \texorpdfstring{$\VL$}{VL}}

$\VNC^1$ is a subtheory of $\VL$,
which in turn is a subtheory
of~$\VSL$ \cite{CookNguyen:book,Kolokolova:thesis}.
As was discussed in
\cite[Sec.\,6.4]{BKKK:Expanders}, $\VNC^1$ has sufficient logic
strength to:
\begin{itemize}
\setlength{\itemsep}{0pt}
\item[\rm (i)] Count the number of members of polynomial size sets; that is, count the
number of elements in a set coded by a string. In particular,
the in-degree and out-degree of a vertex in a graph is definable, concepts
such as the $i$-th outgoing or incoming edge are definable.
\item[\rm (ii)] Reason about integers and rational numbers as represented
by numbers or pairs of numbers.
\item[\rm (iii)] Use strings to encode vectors of integers and vectors
of rational numbers.  The notation $\vec u$ will
be used to denote a vector encoded by a string. That is,
$\vec u$ denotes a second-order object (in spite of using
lowercase $u$.)
\item[\rm (iv)] Define the summation of vectors
of integers and vectors of rational numbers with common denominator.\footnote{%
Jer\'abek~\cite{Jerabek:TCzero} shows that iterated multiplication
is definable in~$\VL$ and thus $\VL$ can
define summations of vectors of rationals without
a common denominator.  This would be good enough
for our purposes, but we prefer to work with the
base theory $\VNC^1$ as much as possible.}
Prove that these summations satisfy properties such as commutativity, associativity,
distributivity,
and invariance under reorderings of terms.

A summation is coded by a string that lists all of
its terms; thus the summation
has length equal to a number (a first-order object).
\item[\rm (v)] Prove the Cauchy--Schwarz theorem. Define the square norm~$\|\vec u\|^2$
of a vector~$\vec u$. Define the projection of~$\vec u$ onto a non-zero vector~$\vec v$
via the formula $(\langle u,v \rangle/\|v\|^2)\vec v$. This is still subject to the proviso
that entries in vectors are rationals with a common denominator.
\item[\rm (vi)]  Use strings to encode the edge relation of a (multi)graph,
and define the rotation map of a graph.
Define edge expansion (see below for more details).
\item[\rm (vii)] Carry out straightforward operations on graphs, such as
adding a self-loop to each vertex, squaring a graph,
or forming 
tensor products of graphs and replacement product of graphs.
\end{itemize}
To this list, we can also add, that the theory $\Vzero$ (and hence $\VNC^1$ and~$\VL$) can
\begin{itemize}
\item[\rm (viii)] Define small (logarithmic) powers of constants.
\end{itemize}
Proving the last point is not entirely trivial: what it means is that,
for $p\in \mathbb N$, the function~$f_p(m) = p^{\lceil \log m \rceil}$
is definable in~$\Vzero$. This is well-known for $p=2$.
By convention, the logarithm is base two,
and $f_2(m)$ is the least power of~$2$
greater than or equal to~$m$. This is well-known to be definable and
provably total in~$\Vzero$. In addition $\lceil \log m \rceil $ is
equal to $|m-1|$ where $|\cdot|$ is the usual length function
of bounded arithmetic.  More generally, for any fixed integer $p\ge 2$,
the function~$f_p$ is definable in~$\Vzero$. This uses standard techniques,
e.g., see Cook-Nguyen~\cite{CookNguyen:book}.

In view of (vii) above, it is clear that $\VNC^1$, and hence $\VL$, can
define the derandomized-squaring product of two graphs. Namely, given
graphs $X$ and $G$ coded by strings, $\VNC^1$~can prove the
existence a string~$Z$ encoding the graph $X \circleS G$, and prove its
straightforward properties; e.g., defining the rotation map of~$Z$ in terms
of the rotation maps of $X$ and~$G$.

$\VNC^1$ can use strings to encode matrices of integers and rationals.
For a matrix~$M$ and vector~$\vec v$ whose entries are rationals
with a common denominator, $\VNC^1$~can define the product $M\vec v$
as well. $\VNC^1$~proves straightforward properties of matrix
products, e.g., associativity and distributivity. $\VNC^1$~can similarly
use strings to encode tensors, and prove straightforward properties
of tensors.

One place where $\VNC^1$ has difficulties is with defining the
norm~$\|\vec v\|$ of~$\vec v$, even if the entries of~$\vec v$
have common denominator. This is because $\VNC^1$~has difficulties working
with square roots, as the square root of a rational number may not be
a rational. However, $\VNC^1$ \emph{can} define $\|\vec v\|^2$ via
summation. For instance, $\VNC^1$~can express the property
$\|\vec v\| = \alpha$ as $\|\vec v\|^2 = \alpha^2$
and the property
$\|\vec v\| \le \alpha$ as $\|\vec v\|^2 \le \alpha^2$.
In this
setting, $\vec v$~is a vector of rational numbers---with common denominator---and
is coded by a second-order object~$V$ (a set) and
$\alpha$~is a first-order object coding a rational number.  The fact that
$\VNC^1$~can express, say, $\|\vec v\|^2 \le \alpha^2$ corresponds to the fact
that this is an alternating logtime property of $v$ and~$\alpha$.

For example, $\VNC^1$ can state and prove the Cauchy--Schwarz theorem in the form 
$\langle \vec x ,\vec y \rangle^2 \le \|x\|^2 \cdot \|\vec y\|^2$~\cite{BKKK:Expanders}.

A matrix norm~$\|M\|$ is harder to express. Let a matrix~$M$ of
rational numbers with common denominator be represented by a second
order object, also denoted~$M$. The statement $\|M\| \le \alpha$ is
definable by
\begin{equation}\label{eq:matrixNorm}
     \forall \vec v \, (\|M \vec v\|^2 \le \alpha^2 \|\vec v\|^2),
\end{equation}
where it is implicit in the notation that $\vec v$ and~$M$ contain
rational numbers with common denominators. Note
that (\ref{eq:matrixNorm}) is a $\forall \Sigma_0$-formula.

Even though $\VNC^1$ can define matrix norms only in a roundabout way,
it can readily prove straightforward properties about vector and matrix norms.
The next lemma has a couple of simple examples that
will be useful later.
\begin{lemma}\label{lem:TriangleInEqMatrixNms}
Working with vectors and matrices that involve rational numbers with common
denominators, $\VNC^1$~can prove:
\begin{itemize}
\item[\rm (a)] (Triangle inequality)
If $\|\vec u\|^2 \le a^2$ and $\|\vec v\|^2 \le b^2$,
then $\|\vec u +\vec v \|^2 \le (a+b)^2$.
\item[\rm (b)] (Matrix product norms)
If $\|M\| \le a$ and $\|N\|\le b$, then $\|MN\| \le ab$.
\end{itemize}
\end{lemma}

\begin{proof}
For~(a), the following argument can be carried out in~$\VNC^1$.
We have
\[
\| \vec u + \vec v\|^2 ~=~ \langle \vec u+\vec v, \vec u + \vec v \rangle
  ~=~ \|\vec u\|^2 + \|\vec v\|^2 + 2\langle \vec u, \vec v \rangle.
\]
Therefore, by the Cauchy--Schwarz inequality,
\begin{equation}\label{eq:CauchySchwarzPf}
(\, \| \vec u + \vec v\|^2 - \|\vec u\|^2 - \|\vec v\|^2 \, )^2
   ~=~ 4 \langle \vec u, \vec v \rangle^2 ~\le~
        4 (\|\vec u \|^2 \cdot \|\vec v\|^2) .
\end{equation}
Suppose that the triangle
inequality fails so $\|\vec u\|\le a^2$
and $\|\vec v\|\le b^2$ and $\|\vec u + \vec v\| > (a+b)^2$. Then the left-hand side of~(\ref{eq:CauchySchwarzPf}) is 
$> ((a+b)^2 -a^2 -b^2)^2 = 4 a^2 b^2$,
and the right-hand side is $\le 4 a^2 b^2$.
This is a contradiction.

The proof of~(b) is entirely trivial. First of all,
for an arbitrary vector~$\vec v$, the bounds on the norms give
$\| M(N\vec v) \|^2 \le a^2\cdot b^2 \cdot \|\vec v\|^2$.
The fact that $M(N\vec v) =(MN)\vec v$ follows immediately since
$\VNC^1$ can prove simple facts about summations including distributivity and
the invariance of summations under reordering of terms.
\end{proof}

The property of $M$ being the adjacency matrix for a
graph~$G$ is straightforward to express in~$\VNC^1$ given that
rotation map (or, equivalently, the edge relation) of~$G$ is encoded
by a second-order object.  For such a matrix, the property of the
mixing ratio being~$\le \alpha$ is definable by
\begin{equation}\label{eq:mixingRatioVNCone}
     \forall \vec v \Bigl[ \vec v \perp \vecone \limplies \|M \vec v\|^2 \le \alpha^2 \|\vec v\|^2 \Bigr].
\end{equation}
Again $\vec v$ is a vector of rational numbers with common denominator and
is encoded by a second-order object~$V$. Let $\MixRat(\vec v, M, \alpha)$ denote
the subformula of~(\ref{eq:mixingRatioVNCone}) enclosed in square brackets, and
let the entire formula~(\ref{eq:mixingRatioVNCone}) be denoted
$\MixRat(G,\alpha)$  (where we are conflating the graph~$G$ with
its adjacency matrix~$M$). Note
that (\ref{eq:mixingRatioVNCone}) is a $\forall \Sigma_0$-formula.

An upper bound on the edge expansion can be expressed by a $\Pi^B_1$-formula
in~$\VNC^1$. Referring back to Definition~\ref{def:edgeExpansion}, the
property that a $d$-regular graph~$G$ has edge expansion~$\ge \alpha$ is expressed by
\begin{equation}\label{eq:edgeExpandVNCone}
\forall U\, \Bigl[ U\subset V(G) \;\land\; 0<|U|\le |V(G)|/2 \;\;\limplies\;\;
     \frac {|E(U, V(G)\setminus U)|}{d\cdot |U|} \ge \alpha \Bigr],
\end{equation}
where ``$V(G)$'' denotes the set of vertices of~$G$,
and $|\cdot|$ denotes set cardinality.  The part of the
formula in square brackets is expressible by
a $\Delta^B_1$-formula and computable in
alternating log time, but the second-order
quantifier~``$\forall U$'' ranging over subsets of vertices
of~$G$ makes (\ref{eq:edgeExpandVNCone}) a $\Pi^B_1$-property.

Let $\EdgeExp( U, G, \alpha)$
respectively denote
the subformula of~(\ref{eq:edgeExpandVNCone})
that is enclosed in square brackets. The entire
formula~(\ref{eq:edgeExpandVNCone}), namely
$\forall U \, \EdgeExp(U, G,\alpha)$, is denoted
$\EdgeExp(G,\alpha)$.

\smallskip

Since $\VNC^1$ can define both mixing ratio and edge expansion,
it can also \emph{state} (a version of) the Cheeger Inequality. 
Unfortunately, it is an open
question whether $\VNC^1$ or $\VL$ can \emph{prove} the Cheeger Inequality.

We will work with the following
form of the Cheeger lemma (for regular, undirected graphs).

\begin{theorem}[Cheeger Lemma]\label{thm:Cheeger}
Let $G$ be a regular undirected graph. Let $\epsilon$ be the
edge expansion of~$G$, and $\deltaGap$~be the spectral gap
of~$G$.  Then
\[
2 \epsilon ~\ge~ \deltaGap ~\ge~ \frac{\epsilon^2}{2}.
\]
\end{theorem}

The Cheeger Lemma cannot be directly stated in this form
when working in the theory~$\VNC^1$. Instead, 
$\VNC^1$~can work only with bounds on
edge expansion and the mixing ratio, since it is not known whether
the values for edge expansion~$\epsilon$ and the spectral gap~$\deltaGap$ or 
the mixing ratio $1-\deltaGap$
are computable in alternating log time.
We can restate the Cheeger lemma for (possible) provability in
$\VNC^1$ as follows.
\begin{theorem}[Cheeger Lemma --- For formalization in~$\VNC^1$; however only part~(a) is
known to be provable in~$\VNC^1$, see Section~\ref{sec:Cheeger}]
\label{thm:CheegerVNCone}
Let $G$ be a d-regular undirected graph, $M$ be the adjacency matrix
for~$G$, and $\alpha >0$ a (rational) number. Then, with $U$~ranging over
sets of vertices of~$G$, and $V$~ranging over vectors of rationals
with common denominator:
\begin{description}
\item[\rm (a)] (Formalizing $\deltaGap \le 2\epsilon$.)
\[
\exists U\, \lnot \EdgeExp(U, G, \alpha)
    ~\limplies~ \exists V\, \lnot \MixRat(V, M, 1{-}2\alpha) .
\]
\item[\rm (b)] (Formalizing $\epsilon^2/2 \le \deltaGap$.)
\[
\exists V\, \lnot \MixRat(V, M, 1{-}\alpha^2/2)
    ~\limplies~ \exists U\, \lnot \EdgeExp(U, G, \alpha) .
\]
\end{description}
\end{theorem}
Theorem~\ref{thm:CheegerVNCone} makes the computational content
of the Cheeger Lemma explicit by using implications between
existential statements. Part~(a) asserts that if there is
a set~$U$ witnessing
that the edge expansion of~$G$ is $<\alpha$, then there is
a set~$V$ that witnesses that mixing ratio is $>1-2\alpha$,
which means that the spectral gap~$\deltaGap$ is $< 2\alpha$.
For $\VNC^1$ to prove~(a), the mapping $U \mapsto V$ must
be computable in alternating log time, with its correctness
provable in~$\VNC^1$.

Part~(b) similarly makes constructive
the other inequality of the Cheeger Lemma. It states that if
there is a vector~$V$ that shows the mixing ratio is $>1-\alpha^2/2$
(so that the spectral gap is $< \alpha^2/2$), there is a set~$U$ witnessing
that the edge expansion is $<\alpha$.

Part~(a) of Theorem~\ref{thm:CheegerVNCone}
is provable in~$\VNC^1$; its proof is postponed until
Section~\ref{sec:Cheeger}. (Furthermore, part~(a) is not used in
the present paper.) It is open whether part~(b) of
Theorem~\ref{thm:CheegerVNCone} is provable in
$\VNC^1$ or $\VL$.  Fortunately, there is a weakened
form of the Cheeger Theorem that applies to graphs~$G$ with 
sufficiently many self-loops
that is known to be provable in~$\VNC^1$:

\begin{theorem}[Cheeger-Mihail Lemma, formalized in~$\VNC^1$,
\cite{Mihail:Markov,BKKK:Expanders}] 
\label{thm:CheegerMihailVNCone}
$\VNC^1$ can prove that if (i)~$d$, $G$, $M$, $\alpha$, $U$ and~$V$ are
as in Theorem~\ref{thm:CheegerVNCone} and (ii)~$d$~is even and every
vertex of~$G$ has at least d/2 self-loops, then:
\begin{description}
\item[\rm (b)] (Formalizing $\epsilon^2/2 \le \deltaGap$.)
\[
\exists V\, \lnot \MixRat(V, M, 1{-}\alpha^2/2)
    ~\limplies~ \exists U\, \lnot \EdgeExp(U, G, \alpha) .
\]
\end{description}
\end{theorem}
\begin{proof}
This proof is based on a construction of Mihail~\cite{Mihail:Markov} that
was shown to be formalizable in $\VNC^1$ as Lemma~12 in~\cite{BKKK:Expanders};
see Section~6.6 of~\cite{BKKK:Expanders}. 
This allows $\VNC^1$ to formalize the following argument.
Let $\vec v \perp \vecone$ be a vector of rationals that witnesses 
that the mixing ratio of~$G$ is $>1-\alpha^2/2$;
i.e., that (\ref{eq:mixingRatioVNCone}) fails.
Let $G^-$~be
obtained from~$G$ by removing $d/2$ self-loops from each vertex.  In the
notation of~\cite{BKKK:Expanders}, $G = \bigcirc G^-$.  
Applying Lemma~12 of~\cite{BKKK:Expanders},\footnote{To 
apply Lemma~12 of~\cite{BKKK:Expanders}, use $k=1$ and 
$\pi = \vec v {+} \vec u = \vec v {+} \vecone$
and $\epsilon = 2\alpha$ and $A=M$. The hypothesis that
$\vec v$ witnesses that $G$~has mixing ratio~$>1-\alpha^2/2$
means
\[
\|M\vec v\|^2 > (1-\alpha^2/2)^2 \|\vec v\|^2 > (1-\epsilon^2 / 4) \|\vec v\|^2 .
\]
From this, Lemma~12 of~\cite{BKKK:Expanders} gives a set~$U$ of vertices
with edge expanson~$<2 \alpha$.} 
there is a set~$U$ of vertices showing that
the edge expansion of~$G^-$ is $< 2\alpha$. Since self-loops
do not contribute edges towards expansion and since the degree of~$G$
is twice that of~$G^-$, the edge expansion of~$G^-$ is twice that of~$G$.
In particular, since $U$ witnesses that the edge expansion of~$G^-$ is
less than~$2\alpha$, then the same~$U$ witnesses that the edge expansion of~$G$ is
less than~$\alpha$.  Item~(b) follows.
\end{proof}

Cheeger's lemma was used in the earlier proof that $\L=\SL$ in two places:
in Theorem~\ref{thm:connectedBis} and indirectly in Theorem~\ref{thm:connectedTri}.
We shall discuss later how to use the 
Cheeger-Mihail Theorem~\ref{thm:CheegerMihailVNCone} instead.
We also need to use the Cheeger-Mihail Theorem 
in~$\VL$ to prove the existence of
$(Q^i,Q,1/100)$-graphs~$H_i$ with good mixing ratio. 
These will be obtained via the
constant degree expander graphs that can be proved to
exist in~$\VNC^1$~\cite{BKKK:Expanders}.

\subsection{Formalizing the statement \texorpdfstring{$\L = \SL$}{L=SL} in \texorpdfstring{$\VL$}{VL}}

Before we talk about proving $\L = \SL$ in~$\VL$,
we need to explain how to formalize the statement that
$\L = \SL$ in~$\VL$.  There are two natural definitions
for an undirected graph to be connected.

\begin{definition}
[``Subset-Connected''] An undirected
graph~$G$ is \emph{subset-connected} provided that for every proper
subset~$U$ of the vertices of~$G$, there is an edge from $U$ to~$\overline U$.
\end{definition}

In $\VNC^1$, subset-connectedness can be expressed as a $\Pi^{1,b}_1$-formula
\[
\forall U\, \bigl[ U\subset V(G) \;\land\; 0 < |U| < |V(G)| \;\;\limplies\;\;
     E(U, V(G)\setminus U) \not= \emptyset \bigr].
\]
This formula is denoted $\SubsetConn(G)$.

\begin{definition}
[``Path-Connected''] As undirected graph is \emph{path-connected}
provided that for every pair of vertices $x$ and~$y$, there is
a path in~$G$ from $x$ to~$y$.
\end{definition}
Path-connectedness can be expressed in~$\VNC^1$ as the statement that
there is a set coding paths joining all pairs of
vertices $x$ and~$y$. Recall that $Z(\cdot,\cdot)$ was
used to denote a second-order object coding a path. We extend this notation by
considering a 4-place predicate, $Z(\cdot,\cdot,\cdot,\cdot)$,
coded as a set, and let $Z^{[x,y]}$ denote the 2-place predicate
such that $Z^{[x,y]}(i,u)$ is equal to $Z(x,y,i,u)$.
Then $\VNC^1$ can express that $G$ is path-connected with a
$\Sigma^{1,b}_1$-formula
\[
\exists Z\,
      \bigl[\forall x, y\in V(G)\,
      (\hbox{``$Z^{[x,y]}$ codes a path from $x$ to $y$''}) \bigr] .
\]
We denote this formula as $\PathConn(G)$.

Clearly, $\SubsetConn(G)$ and $\PathConn(G)$ both express
that $G$ is connected. In one direction $\VNC^1$
proves
\[
\PathConn(G) ~\limplies \SubsetConn(G)
\]
with a simple proof that runs as follows:
Suppose $U$~witnesses $\SubsetConn(G)$
is false and $Z$~witnesses that $\PathConn(G)$ is true.  Letting
$x\in U$ and $y \notin U$, then the path $Z[x,y]$ must have a first
vertex $u \notin U$, and this gives an edge between
$U$ and~$V(G)\setminus U$.

Proving the other direction is harder. In fact, if
\begin{equation}\label{eq:SubsetPathImplyVNC1}
\VNC^1 ~\vdash~ \SubsetConn(G) \limplies \PathConn(G)
\end{equation}
there must be an alternating log time (i.e., uniform~$\NC^1$)
algorithm for determining whether a given graph~$G$ is
connected. To see this, note
that if $\VNC^1$ can prove the $\Sigma^{1,b}_1$-formula
\[
\lnot \SubsetConn(G) \lor \PathConn(G),
\]
then by the witnessing theorem for~$\VNC^1$ (see~\cite{CookNguyen:book}),
there is a set-valued function, computable in alternating log time,
which, provably in~$\VNC^1$, produces a pair $U, Z$
such that either $U$~witnesses
the falsity of~$\SubsetConn(G)$ or $Z$~witnesses the truth of~$\PathConn(G)$.
Checking the truth or falsity of $\SubsetConn(G)$ and
$\PathConn(G)$ given $U,Z$ is in the log time hierarchy, i.e., in uniform~$\AC^0$.
It is conjectured that no such algorithm exists and hence that
$\VNC^1$ does not prove $\SubsetConn(G) \limplies \PathConn(G)$.

On the other hand, $\VL$ does prove $\SubsetConn(G) \limplies \PathConn(G)$
as a consequence of proving $\L=\SL$,
see Theorem~\ref{thm:SubsetPathEquiv}.  The
next theorem formalizes the main result needed
to prove that $\VL\vdash \VSL$.
Here $G$ is a set encoding
an undirected graph, $x$~is a vertex of~$G$, $H$~is a set encoding a subgraph of~$G$,
and $Z$~is a set encoding
paths in~$G$.
\begin{theorem}\label{thm:LeqSLinVL}
$\VL$ proves
\begin{eqnarray*}
\lefteqn{ (\forall \,\hbox{\rm undirected graphs $G$}) (\forall x {\in} V(G)\,)\, \exists H\, \exists Z \, [
    \hbox{\rm ``$H$ is a subgraph of $G$''} \land x\in V(H) } \hspace*{1in} \\
&& \land \, \forall\, u {\in} H\, (
     \hbox{\rm ``$Z^{[u]}$ encodes a path from $x$ to $u$ in~$G$''} )  \\
&& \land \, \lnot\, \exists \,\hbox{\rm an edge}\,\{ u,v\}\in G\,
   (u\in H \land v \notin H) ] .
\end{eqnarray*}
\end{theorem}

The formula of Theorem~\ref{thm:LeqSLinVL} 
is a $\forall \Sigma^B_1$-formula; thus 
its provability in $\VL$ implies that the
functions producing $H,Z$ from $G,x$ are $\Sigma^B_1$-definable
(as logspace computable functions) in~$\VL$.
The condition on~$H$ expresses that $H$~is a maximal connected subgraph containing~$x$.
Thus, $G$~is connected precisely when $H$ is equal to~$G$. This gives:

\begin{theorem}\label{thm:SubsetPathEquiv}
$\VL ~\vdash~ (\forall \,\hbox{\rm undirected graphs $G$})\, [\SubsetConn(G) \liff \PathConn(G)]$.
\end{theorem}

Theorems \ref{thm:LeqSLinVL} and~\ref{thm:SubsetPathEquiv}
are proved in the next section. For the moment, we
observe that Theorem~\ref{thm:LeqSLinVL}
immediately implies our main theorem that $\VL=\VSL$:
\begin{theorem}\label{thm:VSLisVL} 
The theories $\VSL$ and $\VL$ prove the same theorems. That is, $\VSL=\VL$.  
\end{theorem} 

\begin{proof}[Proof of Theorem~\ref{thm:VSLisVL} from Theorem~\ref{thm:LeqSLinVL}]
To prove Theorem~\ref{thm:VSLisVL} we need to show that $\VL$ proves the $\VSL$'s axiom
\[
\VL \vdash \Symm(a,X) ~\limplies~ (\exists C)(\exists Z)\deltaUCONN(a,X,C,Z),
\]
as defined in \eqref{eq:VSLaxiom}.  Here, $C$~is a connected component of $G$ containing the vertex~0, and $Z$~is the set of paths from 0 to all elements of $C$. Now the statement follows from \Cref{thm:LeqSLinVL}  by setting $x=0$. The corresponding $H$  will be  $C$ satisfying $\deltaUCONN(a,X,C,Z)$, and $Z$ will be the set of paths from $0$ to elements of $C$, as required.  
\end{proof}


\subsection{Formalizing the proof of \texorpdfstring{$\L = \SL$}{L=SL} in \texorpdfstring{$\VL$}{VL}}
\label{sec:RVproofVL}

\paragraph*{Preliminaries in $\VNC^1$.}
As a preliminary, we claim that $\VNC^1$ can formalize and prove all
the results in Section~\ref{sec:Prelims}, namely
Theorems \ref{thm:connected}-\ref{thm:JCdecomp}
and replacements for Theorems \ref{thm:connectedBis} and~\ref{thm:connectedTri}.
Theorem~\ref{thm:connected} is 
in terms of subset-connectivity.

\begin{theorem}\label{thm:connectedVNC1}
$\VNC^1$ proves the following: Let $G$ be a
d-regular graph on $n$ vertices. Then
\[
\SubsetConn(G) \;\;\rightarrow \;\;
\EdgeExp(G, 2/(dn)).
\]
\end{theorem}
Theorem~\ref{thm:connectedVNC1} is an immediate consequence
of the definitions since $\VNC^1$ is able to reason about
rational numbers.

Theorem~~\ref{thm:adjacencyTrivial} is easy to state and prove
in $\VNC^1$ since, as discussed above, $\VNC^1$ is able to reason about summations.
The proof of Sedrakyan's Lemma (Lemma~\ref{lem:Sedrakyan})
as given above uses square roots
and thus does not seem to be easily formalizable in~$\VNC^1$. However,
Section~\ref{sec:Sedrakyan} below
gives an alternate proof of Sedrakyan's Lemma that does formalize directly in~$\VNC^1$.
With that in hard, the earlier proof of Theorem~\ref{thm:normM} now formalizes in~$\VNC^1$ since it proceeds by
manipulating summations and invoking Sedrakyan's Lemma.  Note that the proof
of Theorem~\ref{thm:normM} worked by bounding $\|M \vec v\|^2$, so this fits
exactly the way $\VNC^1$ expresses matrix norms.

Theorem~\ref{thm:JCdecomp} and its proof also formalize directly in~$\VNC^1$
in form given in Section~\ref{sec:Prelims}. When formalized,
this theorem has the form ``$\forall M\, \exists C\, \exists D (\cdots)$''.
Here $C$ and~$D$ are easy to define in terms of~$M$.

Theorem~\ref{thm:connectedBis} bounds the mixing ratio of an undirected graph.
To make it provable in~$\VNC^1$, we need the additional hypothesis that
there are sufficiently many self-loops at each vertex so that
Mihail's version of the Cheeger Inequality can be used. 
It becomes:
\begin{theorem}\label{thm:connectedBisVNC1}
$\VNC^1$ proves the following:
Let $G$ be a $d$-regular, undirected graph on $n$ vertices.
Suppose $d$ is even, and $G$ has at least $d/2$ self-loops at
each vertex.
Then 
\[
\SubsetConn(G) \;\;\limplies\;\; \MixRat(G, 1 - 2/(dn)^2).
\]
\end{theorem}

The proof of Theorem~\ref{thm:connectedBisVNC1} in $\VNC^1$ is immediate
from Theorem~\ref{thm:connectedVNC1} and
the Cheeger-Mihail Theorem~\ref{thm:CheegerMihailVNCone}.

Theorem~\ref{thm:connectedTri} also needs to be modified so that 
Theorem~\ref{thm:connectedBisVNC1} can be applied to the matrix~$H$. It
becomes:

\begin{theorem}\label{thm:connectedTriVNC1}
$\VNC^1$ proves the following:
Suppose $d$ is even and that $G$ is a subset connected, $d$-regular, directed graph on $n$ vertices
with a self-loop at each vertex.  Let $M$ be the adjacency matrix of~$G$,
and $H$ be the undirected graph with adjaceny matrix $M^\transpose M$.
Suppose $H$ has at least
$d^2/2$ self-loops at each vertex.  Then the mixing ratio of~$G$ is
at most $1 - 1/(d^4n^2)$.
\end{theorem}

The proof of Theorem~\ref{thm:connectedTriVNC1} in $\VNC^1$
follows the earlier proof of Theorem~\ref{thm:connectedTri}.
The graph~$H$ has degree~$d^2$; therefore Theorem~\ref{thm:connectedBisVNC1}
applies to~$H$.
The rest of the proof of Theorem~\ref{thm:connectedTri} formalizes straightforwardly
in~$\VNC^1$. The only change needed to carry out
the proof in $\VNC^1$ is to talk about the squares
of vector norms.  The last paragraph of the proof 
of Theorem~\ref{thm:connectedTri} is modified so that $\VNC^1$ now
proves 
\[
\| M \vec v \|^2 \le (1-1/(d^4 n^2))^2.
\]
This is done by starting with
\[
\|M^\transpose M \vec v\|^2 \le  (1-\penalty10000 2/(d^4 n^2))^2
\]
from Theorem~\ref{thm:connectedBisVNC1}, and then proving
that 
\[
\| M \vec v \|^4 \le (1-\penalty10000 1/(d^4 n^2))^4.
\]
From that,
\[
\| M \vec v \|^2 \le (1-1/(d^4 n^2))^2
\]
follows immediately.

One final preliminary result is that $\VNC^1$ proves that
$\|P\| = \sqrt n$ and $\|L\| = 1/\sqrt n$, as discussed at the
end of Section~\ref{sec:Prelims}. These statements involve
a square root, but the square roots disappear when formalizing
these in ~$\VNC^1$.  We let $L$ and~$P$ have dimensions $mn\times m$ and
$m \times mn$. To prove that $\|L\|=1/\sqrt n$, $\VNC^1$~proves
\[
\|L \vec v\|^2 = (1/n)\cdot \|\vec v\|^2
\]
holds for all $m$-vectors~$\vec v$.
This follows immediately from the fact that $\sum_{i=1}^n (1/n)^2 = 1/n$
and the fact that $\VNC^1$ can prove this and
reason effectively about summations of rationals.
$\VNC^1$ also proves 
\[
\|P \vec v\|^2 \le n \|\vec v\|^2.
\]
Let $\vec w$ be an $mn$-vector with entries $w_{i,j}$; then
the $i$-th entry of $P\vec w$ is $\sum_{j=1}^n w_{i,j}$.
$\VNC^1$~must prove that 
\[
\bigl( \sum_{j=1}^n w_{i,j} \bigr)^2\leq n \cdot \sum_{j=1}^n w_{i,j}^2.
\]
This follows
by the Cauchy--Schwarz theorem that the
squared dot product $(\vec v \cdot \vecone)^2$
is less than or equal to $\|\vec v\|^2 \cdot \|\vecone\|^2 = \|\vec v\|^2 n$,
with $\vec v$ the $n$-vector with entries $w_{i,j}$, for $j=1,\dots, n$.

\paragraph*{Main part of proof in $\VL$.}
We next discuss how to formalize the proof that $\L=\SL$
in~$\VL$, following closely the exposition in Section~\ref{sec:RVproof}.
Many of the steps can be formalized in $\VNC^1$ in fact,
but some of the crucial steps require the use of~$\VL$.
The starting input is an undirected graph~$Y$ on $N$ vertics without
self-loops or multiedges. In order to let
$\VNC^1$ and $\VL$ reason about graphs, we represent
a graph such as~$Y$ with the string encoding the edge relation
on~$Y$.

\smallskip
\textbf{The first step} in the proof of $\L = \SL$ was 
to transform~$Y$ into a 4-regular directed
graph~$X$. The graph~$X$ had a self-loop
at each vertex. For formalization in $\VL$, this
step is modified to use a 16-regular directed graph~$X^*$ formed
similarly to~$X$ but
with 13 self-loops at each vertex instead of one.\footnote{It would suffice 
for $X^*$ to
have seven self-loops at each vertex in order to apply 
Theorem~\ref{thm:connectedTriVNC1} to~$X^*$ to bound the mixing
ratio of~$X^*$. However,
later when Theorem~32 of~\cite{BKKK:Expanders} is used to
construct the graphs~$H_i$, it is convenient
for the degree of~$X^*$ to be a power
of~2. Hence, $X^*$ is taken to be 16-regular. This makes little
difference to the constants used in the constructions.}
Namely, each vertex in~$X^*$ has the same four edges of~$X$
with labels 1--4 plus 12 additional self-loops with labels 5--16.

This step is completely straightforward to
formalize in~$\VNC^1$: Each vertex~$y$ in~$Y$ is replaced with
$d$ vertices to form~$X^*$, where $d$ is the degree of~$v$.  For $i<d$,
the edges~$\langle y,u\rangle$ incident to~$v$ in~$Y$ are ordered according to
the index~$u$ of the other vertex of the edge. This is easily defined
with counting, and the rest of the construction is straightforward to
carry out in~$\VNC^1$. In particular, the number~$N$ of vertices in~$X^*$
is the sum of the degrees of the vertices in~$Y$, and hence equal to
twice the number of edges in~$Y$. Let $E(\cdot,\cdot)$ be the
edge relation for~$Y$. For a vertex $y$ of~$Y$, the $i$-th neighbor of~$y$
is equal the value~$u$ (if any) such that $E(y,u)$ holds and
there are $i{-}1$ many values~$u^\prime<u$ for which $E(y,u^\prime)$ holds.
The vertex~$y$ of~$Y$ is replaced in~$X$ with the vertices $u = j+k$
where $j$~is the sum of the degrees of the vertices $y^\prime<y$ in~$Y$ (namely
the number of pairs $\langle u^\prime, u^\pprime\rangle$ such that $E(u^\prime,u^\pprime)$
with $u^\prime < y$) and where $k$~is less than the degree of~$y$ in~$Y$.
Each vertex~$x$ in~$X^*$ has degree~$16$. Given $i<16$, it is straightforward
to define, in~$\VNC^1$, the $i$-th neighbor of~$x$ in~$X^*$.

The self-loops in~$X^*$ allow $\VNC^1$
to apply Theorem~\ref{thm:connectedTriVNC1} to~$X^*$.
Form the undirected graph~$H$ by letting $M$~be the
adjacency matrix of~$X^*$ and letting $M^\transpose M$ be
the adjacency matrix of~$H$.  Then $H$ is 256-regular.
Furthermore, $H$~has 172 self-loops at each vertex $\langle v, i \rangle$. This is because
the thirteen self-loops at a vertex in~$X^*$ contribute $13^2 = 169$
self-loops in~$H$ at that vertex, plus the each of the three non-loop outgoing edges 
from $\langle v,i\rangle$ in~$X^*$  separately contributes
one more self-loop to~$\langle v, i \rangle$ in~$H$.
Since $172 \ge 256/2$, $X^*$~and $H$ satisfy the hypotheses of 
Theorem~\ref{thm:connectedTriVNC1}, 
thus $X^*$~has mixing
ratio at most $1 - 1/(16^4 N^2)$ where $N$ is the number of vertices of~$X^*$.
This is provable in~$\VNC^1$.

As discussed earlier, the Rozenman-Vadhan proof
of $\L=\SL$ used the constant $Q=4^q$ for some constant~$q$ since $X$ had
degree~4 (see Figure~\ref{fig:XGdegrees}). The $\VL$-proof uses the 16-regular~$X^*$
instead of~$X$.
For this reason, we now redefine~$Q$ to equal ${16}^q$ instead instead of~$4^q$. 
The value of~$q$ is defined in the next paragraphs so that 
Theorem~32 of~\cite{BKKK:Expanders} combined with Theorem~\ref{thm:connectedTriVNC1} above
will give large $Q$-regular graphs~$H_i$ with
good mixing ratio.\footnote{Theorem~32 of~\cite{BKKK:Expanders} holds
when $d=2^\ell$ is any sufficiently large power of~2.}

\smallskip
\textbf{The second step} is to prove the existence of the graphs~$H_i$
with the desired expansion properties.  
For $i > 0$, $H_i$~is to be a $(Q^i, Q, 1/100)$-graph.   
By \cite[Theorem 32]{BKKK:Expanders}, 
$\VNC^1$ proves that for some constant~$p$,
letting $d = 2^p$, there are undirected $d$-regular graphs of
arbitrary size
which have edge expansion $>1/2592$.  By replacing undirected edges
with pairs of directed edges, there are also $d$-regular
directed graphs with the same properties.
Let $q'$ be the least value such that $16^{q'} = 2^{4q'} \ge 4\cdot 2^p$,
and let $Q' = 16^{q'}$. Of course $4q' - p \in \{2,3,4,5\}$.

Fix $k>0$ a constant, to be defined in the next paragraph.
Let $Q = (Q')^k$, and let $q = q' \cdot k$, so $Q = 16 ^ q$.
Letting $i=O(\log N)$, we claim that $\VNC^1$
can prove the existence of an undirected $Q'$-regular graph~$K_i$
on $Q^i$~many vertices with nontrivial mixing ratio.  
By Theorem~32 of~\cite{BKKK:Expanders}, $\VNC^1$ can prove
that there is an undirected
$2^p$-regular graph~$K'_i$ on $Q^i$~many vertices
with edge expansion $\ge 1/2592$. 
Form the $Q'$-regular undirected graph~$K_i$ by 
adding $16^{q'} - 2^p$ many self-loops to each vertex in~$K'_i$.
This increases the degree of the graph by a factor
$f = 16^{q'}/2^p$, and $4 \le f \le 32$.  Therefore,
$K_i$~has edge expansion $\ge 1/(32\cdot 2592)$.
Since at least $3/4$ of the edges on any vertex on~$K_i$ are
self-loops and since $(3/4)^2 \ge 1/2$, 
Theorem~\ref{thm:connectedTriVNC1} applies to~$K_i$ and implies
that $K_i$ has mixing ratio $\le 1-\delta$ where
$\delta = (1/(32\cdot 2592))^2/2$, still provably in $\VNC^1$.
Thus the value~$\delta$ is a lower bound on the \emph{spectral gap} of~$K_i$.

We choose $k$ to be the smallest integer such that $(1-\delta)^k \le 1/100$.
Let $H_i = K_i^k$. Thus $H_i$~is $Q$-regular (since $Q = (Q')^k$)
and still has $Q^i$ many vertices.
In addition, $\VNC^1$ can
prove that the mixing ratio of~$H_i$ at most $(1-\delta)^k \le 1/100$.  This is
proved in $\VNC^1$ using the fact that, for any~$\vec v \perp \vecone$,
we have $(K_i \vec v) \perp \vecone$ and $\|K_i \vec v\| \le (1-\delta)\|\vec v\|$.
We have established
that $\VNC^1$ can prove that $H_i$ exists and is a
$(Q^i,Q,1/100)$ graph.  

The above construction justified that
$\VNC^1$ can prove the existence of a
(uniquely specified)~$H_i$ as a function of~$i$, so the function that maps a
number~$i$ to a string encoding~$H_i$
is $\Sigma^B_1$-definable function
in~$\VNC^1$. This construction can be
parallelized, and since $m_1 = O(log N)$
is a number (a first-order object),
$\VNC^1$ can prove the existence
a string that simultaneously encodes
all of the graphs~$H_i$ for $i<m_1$, as a function of~$N$.
This completes the second step.

The second step of the proof as formalized in~$\VNC^1$
used different values for the constants $q$ and $\delta$ than were
used in Section~\ref{sec:LSLproof}. This will need to be taken into account
below when showing that $\VL$ can prove the new version
of Claim~\ref{claim:lambdaX}.

\smallskip
\textbf{The third step} is to define the graphs~$G_i$
and prove they have the desired expansion properties.
For $i\le m_0$, we set $G_i = H_i$. For the remaining matrices,
we set \[
G_{m_0+i} = (H_{m_0+2^i-1})^{2^i},
\]
where $i=\log\log N + O(1)$.
Of course, $G_{m_0+i}$ has $Q^{m_0+2^i-1}$ vertices and is
$Q^{2^i}$-regular.  Since $i = \log\log N + O(1)$ and $Q = O(1)$,
these facts are readily expressed and proved by~$\VL$.
In particular, the number of vertices $Q^{m_0+2^i-1}$
and the degree $Q^{2^i}$ are both $N^{O(1)}$ and thus
are numbers (first-order objects).
This allows an edge label~$e$ in~$G_{m_0+i}$ to be
specified by a sequence $\langle a_1,\ldots, a_{2^i} \rangle$
with each $a_i<Q$, namely as a sequence of $2^i$ steps
in~$H_{m_0+2^i-1}$.  Thus $\VL$ can formalize all these concepts.
Indeed, $\VL$~can define the function that takes
a value~$i\le m_1-m_0$, a vertex~$w$ of~$G_{m_0+i}$,
and an edge label~$e<Q^{2^i}$ and produces the vertex~$u$
of~$G_{m_0+i}$ which is the $e$-th neighbor of~$w$ in~$G_{m_0+i}$.

This allows $\VL$ to define the entries in the adjacency matrix~$M_{G_{m_0+i}}$
of $G_{m_0+i}$. Namely the $(w,u)$-entry of the unnormalized
adjacency matrix is the number of edges from $w$ to~$u$ in $G_{m_0+i}$;
the same entry in the (normalized) adjacency matrix is
this number divided by~$Q^{2^i}$.

We furthermore claim that $\VL$ can
prove that the mixing ratio of $M_{G_{m_0+i}}$ is $<(1/100)^{2^i}$.
The idea for the $\VL$ proof is that it establishes
that if $\vec v \perp \vecone$ with $\|\vec v\|=1$ then
\begin{equation}\label{eq:mixingViGm0i}
(M_{H_{m_0+2^i-1}})^j(\vec v) \perp \vecone \qquad \hbox{and} \qquad
\|(M_{H_{m_0+2^i-1}})^j(\vec v) \| < (1/100)^j,
\end{equation}
where $M_{H_{m_0+2^i-1}}$ is the adjacency matrix of
$H_{m_0+2^i-1}$,
by using induction with $j$ ranging from $1$ to~$2^i$.  $\VL$~can formalize
this kind of induction \emph{provided} that it can meaningfully define
the matrices~$M_j$ and vectors~$\vec v_j$
\[
M_j ~:=~ (M_{H_{m_0+2^i-1}})^j
\qquad\hbox{and}\qquad
\vec v_j ~:=~ (M_{H_{m_0+2^i-1}})^j(\vec v).
\]
These two concepts were already shown to be $\VL$-definable in the
previous two paragraphs when $j=2^i$ for $i=\log\log N + O(1)$:
Thus they
are definable for arbitrary $j=O(\log N)$.  Furthermore,
$\VL$ proves that
\begin{equation}\label{eq:Mip1M1Mi}
M_{j+1} ~=~  M_1 \cdot M_j.
\end{equation}
It does this by showing that the number of edges in~$M_{j+1}$
from vertex $w$ to~$u$ is equal to
\begin{eqnarray*}
\lefteqn{\sum_x  ~~
     [(\hbox{\# of edges from $w$ to $x$ in $(M_{H_{m_0+2^i-1}})^j$})} \\
&&\qquad\qquad
\cdot (\hbox{\# of edges from $x$ to $u$ in~$M_{H_{m_0+2^i-1}}$}) ],
\end{eqnarray*}
where $x$ ranges over vertices (from $[N]$).  This is just another
way of stating the matrix product identity~(\ref{eq:Mip1M1Mi}).
Given this, $\VL$~readily proves (\ref{eq:mixingViGm0i}) by induction
on~$j$.\footnote{This uses induction only to a logarithmically sized
number $j = O(\log N)$. In fact, $\VL$~can use induction to
an arbitrary integer. The limiting factor is that the property to be proved
must be $\Delta^B_1$. This is why we need $j = O(\log N)$; namely,
so that the values $v_j$ and $M_j$ are given by $\Sigma^B_1$-definable functions
in~$\VL$.}

\smallskip
\textbf{The fourth step} is to prove the existence
of the inductively defined graphs~$X_{m_1}$
based on the recurrence $X_{i+1} = X_i \circleS G_i$
of Equation~(\ref{eq:XiDefn}).  This is formalized in~$\VL$
by using Algorithm~\ref{alg:XmEdge} to define the edge relation
of~$X_i$.  The graphs~$X_i$ all have the vertex set~$[N]$.
Also, $X_i$~is directed and $d_i$-regular, with $d_i$ as shown
in Figure~\ref{fig:XGdegrees}; namely, $X_i$~has degree $d_i:=Q_i$
for $i\le m_0$; and for $i=m_0+j$, $X_i$~has degree
$d_i:=Q^{m_0+2^j-1}$.
Similarly to before
(but now generalized to~$X_i$ instead of $X_{m_0}$),
an edge in~$X_i$ is labeled by a tuple~$w$
of length~$i$:
\begin{equation}\label{eq:wInput}
(z,a_1,\ldots,a_{i-1}),
\end{equation}
where
\begin{itemize}
\setlength{\itemsep}{0pt}
\setlength{\parsep}{0pt}
\item[\rm (a)] $z\in[Q]$
\item[\rm (b)] for $j\le m_0$, $a_j\in [Q]$, and
\item[\rm (c)] for $j = m_0 + j^\prime$ (if there is any such $j < i$), $a_j \in [Q^{2^{j^\prime}}]$.
\end{itemize}
$\VL$ can $\Sigma^B_1$-define the functions 
and $\EdgeLabel(b,w)$ and $\Neighbor(x,w)$ by
proving that 
Algorithms \ref{alg:XmEdge-b} and~\ref{alg:XmEdge} 
can be straightforwardly computed in logspace.
$\VL$ first proves that, for any $i\le m_1$, any edge label~$w$ of the
form~(\ref{eq:wInput}) and any vertex $x$ of~$X_i$,
there is a second-order object~$Z = Z(x,w)$ encoding the
entire run of the Algorithm~\ref{alg:XmEdge}
for $\Neighbor$, including the entire runs
of the all of the $2^{i-1}$ calls to
Algorithm~\ref{alg:XmEdge-b} for $\EdgeLabel$.
Note that the $\Sigma^B_1$-definition of $Z(x,w)$ 
is in $\VL$, not $\VNC^1$, because of the sequential
nature of the two algorithms.  $\VL$~then uses induction
on~$i$ to prove that $Z(x,w)$ computes the edge relation
correctly. For this, it uses induction
on~$i$ to prove that
\begin{align*}
(\forall w = (z,a_1,\ldots,a_{i-1}))\, (\forall x \in X_i)\;
 [ \hbox{the}&~\hbox{output~$y$ of the computation $Z(x,w)$} \\ 
 & \hbox{is the $w$-th neighbor of $x$ in~$X_i$}]
\end{align*}
holds for $i\le m_1$.  The quantifiers on $w$ and~$x$ are bounded first-order
quantifiers, so this uses induction on a predicate
which is $\Delta^B_1$ relative to the theory~$\VL$, and hence the induction
can be carried out in~$\VL$.  The proof for the induction step
is formalizable in~$\VL$ since for given values of $i,w,x$,
the transcript $Z(x,w)$ for Algorithm~\ref{alg:XmEdge} consists
of two parts, one for $b<2^{i-2}$ and one for $b\ge 2^{i-2}$, that
by the induction hypothesis correctly compute the edge relations
in~$X_{i-1}$, and since $\VNC^1$, and hence~$\VL$, proves
$X_i = X_{i-1}\circleS G_i$.

\smallskip
\textbf{The fifth step} is to prove good upper bounds
for the mixing ratios of the~$X_i$'s.  This requires a little
care to formalize in~$\VL$.  The first observation is
that the proof of Theorem~\ref{thm:RVexpand} can be
adequately carried out
in~$\VNC^1$; specifically, the proof of part~(b) of
the theorem.  For this, we
claim that $\VL$~proves that if $A$ and~$M$
are the adjacency matrices for $X_i$ and~$X_{i+1}$, then
\begin{equation}\label{eq:PerpProperty}
\vec u \perp \vecone \;\;\rightarrow\;\; A \vec u \perp \vecone
\end{equation}
and
\begin{equation}\label{eq:RVexpandProperty}
\vec u \perp \vecone \; \land\;  \|M\vec u\| > f(\lambda,\mu)\|\vec v\|
  \;\; \rightarrow\;\;
  ( \|A\vec u\| > \lambda \|\vec v\|
        \;\lor\; \|A(A\vec u)\| > \lambda \|A\vec u\| ).
\end{equation}
The vector~$\vec u$ will always be expressible
as a vector of rational numbers with a common
denominator. The entries in the adjacency matrices~$X_i$
are also rational numbers with common denominator~$d_i$,
the degree of~$X_i$.  These denominators will be
integers, i.e., first-order objects.
Therefore the implication~(\ref{eq:PerpProperty})
can be expressed in terms of exact arithmetic on
rational numbers, and the implication~(\ref{eq:PerpProperty})
can be proved in the theory~$\VNC^1$ using the arguments from
Section~\ref{sec:Prelims} used to prove Theorem~\ref{thm:adjacencyTrivial}.
Thus, this is also provable in~$\VL$.

Handling the implication~(\ref{eq:RVexpandProperty}) is a little
more difficult because the proof of Theorem~\ref{thm:RVexpand} reasons about
mixing ratios, and thus needs to reason about vector norms.
Computing a vector norm~$\|\vec v\|$ requires a square root, and this is
a little problematic since it uses a square root and this
takes us outside the domain of rational numbers.
As discussed earlier,
the solution is to argue about \emph{squares of vector norms},~$\|\vec v\|^2$,
whenever possible.  That is, $\VNC^1$ (and hence~$\VL$) will prove
\begin{equation}\label{eq:RVexpandPropertySq}
\vec v \perp \vecone \;\land\; \|M\vec v\|^2 > f(\lambda,\mu)^2 \|\vec v\|^2
  \;\; \rightarrow \;\;
  ( \|A\vec v\|^2 > \lambda^2 \|\vec v\|^2
        \;\lor\; \|A(A\vec v)\|^2 > \lambda^2 \|A\vec v\|^2 )
\end{equation}
instead of proving~(\ref{eq:RVexpandProperty}). In
fact, (\ref{eq:RVexpandPropertySq})~is provable in~$\VNC^1$.

For this, we need to show that the argument in the final paragraph
of the proof of Theorem~\ref{thm:RVexpand} can be formalized
in~$\VNC^1$. 
That is, $\VNC^1$ proves that if 
\[
     \| v \| =1 \hbox{~~and~~} \vec v \perp \vecone \hbox{~~and~~}\| M \vec v \|^2 > f(\lambda, \mu)^2
\]
then either 
\[
\|A \vec v\|^2 > \lambda^2 \|\vec v\|^2
\qquad \hbox {or} \qquad
\|A^2 \vec v\|^2 > \lambda^2 \|A \vec v\|^2.
\]
It is enough to show that $\VNC^1$ proves that if 
$\| M \vec v \|^2 > f(\lambda, \mu)^2$ then $\|A^2 \vec v\|^2 > \lambda^4 $.
Recall that $\VNC^1$ can prove that $\|P\|\le \sqrt n$
and $\|L\| \le 1/\sqrt n$. Also, $\|\tilde A\| \le 1$ since it is a
permutation matrix and $\|I_N \otimes C\| \le 1$ by
Theorem~\ref{thm:JCdecomp}(a). Therefore,
letting 
\[
R = \mu P \tilde A ( I_N \otimes\penalty10000 C ) \tilde A L,
\]
we have $\|R\|\le \mu$ by Lemma~\ref{lem:TriangleInEqMatrixNms}(b).
We will apply the triangle inequality
to the vectors $\vec x = (1-\mu)A^2 \vec v$ and $y = R \vec v$, and scalars
$a = (1-\mu)\lambda^2 \|\vec v\|$ and $b = \mu \|\vec v\|$.
By hypothesis, $M\vec v = \vec x + \vec y$
and
\[
\|M \vec v\|^2 ~=~ \|\vec x + \vec y\|^2
  ~>~ f(\lambda,\mu)^2 \|\vec v\|^2 
  ~=~ (\mu + (1-\mu)\lambda^2)^2\|\vec v\|^2
  ~=~ (a+b)^2.
\]
Since $\|R\|^2\le \mu^2$, we have $\| \vec y \|^2 \le b^2$. 
Assuming (for sake of a contradiction) that $\|A^2 \vec v\|^2 \le \lambda^4 $,
we have
\[
\|\vec x\|^2 ~=~ \|(1-\penalty10000 \mu) A^2 \vec v\|^2  
    ~\le~ a^2\\
    ~=~(1-\mu)^2\lambda^4 \|v\|^2.
\]
This contradicts the triangle inequality of 
Lemma~\ref{lem:TriangleInEqMatrixNms}(a).
Therefore, $\VNC^1$ can prove~(\ref{eq:RVexpandPropertySq}).

\smallskip
For \textbf{the sixth step}, to finish proving upper bounds on the mixing ratios of
the graphs~$X_i$, $\VL$~must prove Proposition~\ref{prop:BoundsOnF}
and Claim~\ref{claim:lambdaX}.  $\VNC^1$~can carry out the proof
of Proposition~\ref{prop:BoundsOnF} as given above.  However,
the proof of Claim~\ref{claim:lambdaX} needs to be carried out in~$\VL$
instead of $\VNC^1$ for
the simple reason that $\VL$ can define and prove the existence of
the graphs~$X_i$ and their expansion properties, and $\VNC^1$~cannot (as
far as we know). As already discussed, $\VL$~can define
the graphs~$G_i$ and $X_i$, can prove $X_{i+1} = X_i \circleS G_i$,
can prove Theorem~\ref{thm:RVexpand}, can define the
functions $f_p(m) = p^{\lceil \log m \rceil}$, and can
reason about rational numbers. The $\VL$-proof of the Claim~\ref{claim:lambdaX}
uses all these facts, but also needs to handle the use of
logarithms.  Logarithms take us outside rational numbers of course, but
$\VL$ (and $\VNC^1$) can instead use rough approximations to logarithms
(specifically, rough approximations to $\lceil \log 3/2 \rceil$ and
$\lceil \log 8/7 \rceil$) as shown by the next very simple lemma.
The constant $16^4$ in the lemma is not optimal. 

\begin{lemma} \label{lem:ineqsVNC1}
Let $N>1$.
\begin{itemize}
\setlength{\itemsep}{0pt}
\item[\rm (a)] $\VNC^1$ proves $16^4 N^2 < (3/2)^{m_0}$
    where $m_0 = 100 \cdot \lceil \log N \rceil$. 
\item[\rm (b)] $\VNC^1$ proves $N^2 < (8/7)^{2^\ell}$
where $\ell = 10 + \lceil \log \log N  \rceil$.
\end{itemize}
\end{lemma}
\begin{proof}
We argue informally in~$\VNC^1$. Since $(3/2)^2 > 2$, we have
\[
(3/2)^{m_0} ~>~ 2^{50 \lceil \log N \rceil}
   ~\ge~ 2^{25 \lceil \log N^2 \rceil}
   ~=~ (2^{\lceil \log N^2 \rceil})^{25}
   ~\ge~ (N^2)^{25} ~>~ 16^4 N^2
\]
for $N \ge 2$.
That proves~(a).  Now, $(8/7)^8 > 2$. Thus
\[
(8/7)^{2^\ell} ~=~ (8/7)^{2^{10 + \lceil \log ( \lceil \log N \rceil) \rceil} }
   ~\ge~ (8/7)^{2^{10} \lceil \log N \rceil}
   ~>~ 2^{2^7 \lceil \log N \rceil} ~\ge~ N^{2^7} ~\ge~ N^2,
\]
proving~(b).\footnote{It is clear from these calculations that the constants 100 and
10 are far from optimal. We have kept them unchanged from the
conventions of~\cite{RozenmanVadhan:DerandSquaring}.}
\end{proof}
Now we need to
show that $\VL$ can prove Claim~\ref{claim:lambdaX} and its bounds
on mixing ratios. That is, $\VL$ proves that if $\lambda(X_1) \le 1- 1/(16^4 N^2)$,
then $\lambda(X_{m_1})^2 < 1/N^4$.  This requires
showing that $\VL$ proves the bounds on the mixing ratios
established in the proof of Claim~\ref{claim:lambdaX}.  That is,
$\VL$ can prove (recall that $\VL$ actually expresses
bounds on the mixing ratios in terms of their squares, as
in Equation~(\ref{eq:mixingRatioVNCone})):
\begin{lemma}\label{lem:mixingInductionPf}
$\VL$ proves
\begin{itemize}
\item[(a)] For $i\le m_0$,
$\lambda(X_i) \le \lambda'_i$ where $\lambda'_i = 1 - \min\{ 1/4, 1{-}(3/2)^i/(256 N^2) \}$, and
\item[(b)] For $m_0 < i \le m_1$,
$\lambda(X_i) \le \lambda'_i$ where $\lambda'_i  = \lambda_{i-m_0} = (64/65)\cdot(7/8)^{2^i}$.
\end{itemize}
\end{lemma}
The conditions $\lambda(X_i) \le \lambda'_i$ are $\forall \Delta^B_1$ since they include
an unbounded first-order quantifier, 
see~(\ref{eq:mixingRatioVNCone}), and
thus $\VL$ cannot use induction directly to prove them. Instead, we
exploit the fact that Theorem~\ref{thm:RVexpand} about
the mixing ratio of $X \circleS G$ shows
that if a vector~$\vec v$ witnesses a large value for the mixing ratio for the adjacency matrix~$M$ of~$X \circleS G$,
then one of $\vec v$ or $A \vec v$ witnesses a large value
for the mixing ratio the adjacency matrix~$A$ of~$X$. 
Let $M_i$ denote the adjacency matrix for~$X_i$. 
From the formalized Lemma~\ref{thm:RVexpand}, $\VL$~proves that if
$\| M_{i+1} \vec v \| > \lambda'_{i+1} \|\vec v\|$,
then 
$\| M_i \vec u \| > \lambda'_i \|\vec u\|$ for at least one
of $\vec u = \vec v$ or $\vec u = M_i \vec v$.

Therefore $\VL$ can argue as follows to prove Lemma~\ref{lem:mixingInductionPf}.  Suppose $1\le j\le i \le m_1$ and that $\|M_i \vec u\| > \lambda'_i \|\vec u\|$.  For $c$ a number,
let $c_k$~be the $k$-th bit in the binary expansion of~$c$, and define $\vec u_{c,j}$ inductively by letting $\vec u_{c,i} = \vec u$
and letting $\vec u_{c,j-1}$ be
\[
\vec u_{c,j-1} ~=~ \left\{\begin{array}{ll}
\vec u_{c,j} ~~~~ & \hbox{if the $j$-bit $c_j$ of $c$ equals 0}  \\
M_{j-1} \vec u_{c,j} & \hbox{otherwise}
\end{array} \right.
\]
Then, from the hypothesis that $\|M_i \vec u\| > \lambda'_i \|\vec u\|$,
$\VL$ uses induction on $j=i,\ldots,1$ to prove that
\[
(\exists c \le 2^i)\; \| M_j \vec u_{c,j}\| > \lambda'_j \| \vec u_{c,j} \| 
\]
holds for all $j\le i$. The induction step
uses the fact that if 
$\| M_{j+1} \vec u_{c,{j+1}}\| > \lambda'_{j+1} \| \vec u_{c_{j+1}} \|$, then
keeping the bits~$c_k$ for $k>j$ fixed, there is some choice for the
value of the $j$-th bit of $c$ for which
$\| M_j \vec u_{c,j}\| > \lambda'_j \| \vec u_{c_j} \|$.
The value
of $\vec u_{c,j}$ is $\Sigma^B_1$-definable,
and $\exists c {\le} 2^i$ is
a bounded first-order quantifier. Therefore this induction is permissible in~$\VL$. 

With $i=m_1$
and $j=1$, $\VL$ proves the needed upper bound 
on the mixing ratio of~$X_{m_1}$. 

\smallskip
\textbf{The seventh step} is to prove Theorem~\ref{thm:expandImpliesConnected}
in~$\VNC^1$. Fortunately, the proof as given earlier
formalizes directly in~$\VNC^1$.

\bigskip

We are now ready to complete the proof that $\VL$ proves
$\L = \SL$. As a first step, we prove Theorem~\ref{thm:SubsetPathEquiv}.

\begin{proof}[Proof of Theorem \ref{thm:SubsetPathEquiv}]
As noted earlier, $\VL$~can easily prove
$\PathConn(G)$ implies $\SubsetConn(G)$.  For the converse,
we argue informally in~$\VL$. Let $Y=G$~be an undirected
graph such that $\SubsetConn(Y)$. Form $X_1$ as a
16-regular, directed graph using the
construction of~$X^*$ used earlier.
Clearly $X_1$ is also subset-connected. By Theorem~\ref{thm:connectedBisVNC1},
$X_1$ has mixing ratio at most $1-1/(16^4 N^2)$.
Using the
constructions above, form~$X_{m_1}$ from~$X_1$.  By the arguments in
Section~\ref{sec:RVproofVL}, which are formalizable in~$\VL$,
$X_{m_1}$~is connected, with each pair of
vertices in~$X_{m_1}$ joined by at least one edge.
Furthermore, unwinding Algorithm~\ref{alg:XmEdge},
each edge~$\langle x, y\rangle$ in~$X_{m_1}$ corresponds to a directed path
of length~$2^{m_1}$ in~$X_1$. By the construction of~$X_1$ from~$Y$, this
path in~$X_1$ yields a path in~$Y$. Therefore $Y$ is path-connected.
\end{proof}

This proof established a stronger property:

\begin{lemma}\label{lem:SubsetConnImpliesEdges}
The theory $\VL$ proves the following.
Suppose $Y$ is an undirected graph such that $\SubsetConn(Y)$.
Let $X_1$ and~$X_{m_1}$ be formed
from~$Y$ as in the above constructions. Then, for every pair of vertices
$x$ and~$y$ in~$X_1$, there is an edge from $x$ to~$y$ in~$X_{m_1}$.
\end{lemma}

We further improve on this as follows:

\begin{theorem}\label{thm:LequalSLmethodVL}
The theory $\VL$ proves the following.
Suppose $Y$ is an undirected graph, and $X_1$ and~$X_{m_1}$ are formed
from~$Y$ as in the above constructions. Then
for each pair $x,y$ of vertices of~$X_1$, there is a directed path
from $x$ to~$y$ in~$X_1$ if and only if there is an edge
from $x$ to~$y$ in~$X_{m_1}$.
\end{theorem}

The proof of the theorem will depend on the fact that Algorithm~\ref{alg:XmEdge}
implements a universal traversal procedure
(see again Rozenman-Vadhan~\cite{RozenmanVadhan:DerandSquaring}). What this means is
that, given a vertex~$x$ in a degree~$D$ directed graph~$X$
and an edge label $w\in X_{m_1}$,
Algorithm~\ref{alg:XmEdge} finds the vertex reached on the $w$-th outgoing
edge from~$x$ in~$X_{m_1}$ by calculating a sequence of edge
labels
\begin{equation}\label{eq:AlgIndices}
      i_1, \, i_2, \, \ldots, \, i_{q 2^{m_1}}
\end{equation}
in~$X$ (so $i_j < Q$)
and traversing the path in~$X$ containing the vertices $x_0,x_1,x_2,\ldots, x_{q 2^{m_1}}$
where $x_0 = x$ and each $x_{j+1} = x_j[i_j]$.  The extra factor of~$q$ arises
because $X_1 = X^q$, so that each edge in~$X_1$ is a sequence of $q$ many edges
in~$X$.

It is important that the traversal sequence of
edge labels~(\ref{eq:AlgIndices}) depends on $w$ and the graphs~$G_i$ (and thereby
depends on the degree~$D$ and size~$N$ of~$X$); however it
does \emph{not} depend on any properties of~$X$ other than its
degree and size.  In fact, the graphs~$G_i$ have only a minimal
dependence on the number~$N$ of vertices in~$X$: the value of~$N$
affects the construction of the~$G_i$'s only
in that is was used to pick the values $m_0$ and~$m_1$. All that is
required is that the number of vertices of the (connected)
graph~$X$ is \emph{at most}~$N$.

\begin{proof}[Proof of Theorem~\ref{thm:LequalSLmethodVL}]
From the above, it is clear that if there is an edge from $x$ to~$y$
in~$X_{m_1}$, then there is a path from $x$ to~$y$ in~$X$, and this is
provable in~$\VL$.
So suppose there is a path from $x$ to~$y$ in~$X$,
we must show that there is an edge from $x$ to~$y$ in~$X_m$, using
arguments that can be formalized in~$\VL$.

Let $U$ be the set of vertices~$z$ that are reachable in~$X$
using a path that is implicitly traversed by an edge~$w$ of~$X_{m_1}$.
Specifically, for each edge label~$w$, use Algorithm~\ref{alg:XmEdge}
to form the edge labels~$i_j$ in~$X$ and
define $x_i[i_j]$ as above, see~(\ref{eq:AlgIndices}).  Then $U$ is the set
of vertices of the form $x_j[i_j]$ for some edge label~$w$ of~$X_{m_1}$ and
some $j\le q 2^{m_1}$. It is possible that $U$ is equal to all of~$X$. If so,
since $U$ is clearly path-connected, it is also subset-connected. Thus, in this
case, Theorem~\ref{thm:LequalSLmethodVL} follows
from Lemma~\ref{lem:SubsetConnImpliesEdges}.

So, suppose instead that $U$ is not all of~$X$. The argument now splits
into two cases depending on whether there is edge from $U$ to~$\overline U$,
that is, depending on whether there are $z\in U$ and $z'\notin U$ joined by an edge.
If there is no such edge, $U$~induces a $4$-regular subgraph of~$X$, which
is path-connected and hence subset-connected. Then $|U|<N$, so the
edge labels $w= ( z, a_1,\ldots, a_{m_1} )$
form a complete traversal sequence for~$U$. As a result, for every pair of
vertices $z,z' \in U$ there is an edge in $X_{m_1}$ from $z$ to~$z'$. In particular,
if there is a path from $x$ to~$y$ in~$X$, then $y\in U$ and hence
there is an edge in~$X_{m_1}$ from $x$ to~$y$.  (This subcase is just
a generalization of the case $U=X$.)

The remaining case to consider is where
there is an edge from $U$ to~$\overline U$.
We prove that this leads to a contradiction and so cannot happen.
Modify the graph~$X$
to form a new graph~$X^\prime$ that satisfies
\begin{enumerate}
\setlength{\itemsep}{0pt}
\item[i.] $X^\prime$ is 4-regular and has the same vertices as~$X$, and
\item[ii.] Any edge in~$U$ joining two vertices in~$X$ is also in~$X^\prime$.
Specifically, if $z, z' \in U$ and the $i$-outgoing edge of~$z$ is the
$i'$-th incoming edge of~$z'$, then the same holds in~$X'$.
\item[iii.] $X^\prime$ is path-connected and hence subset-connected.
\end{enumerate}
Form $X^\prime$ from~$X$ as follows. First remove all edges from~$X$ that
touch a vertex not in~$U$; i.e., remove all edges involving vertices in~$\overline U$.
This removes at least one outgoing edge from some vertex~$u_1$ in~$U$
and (by 4-regularity) at least one incoming edge from some vertex~$u_2$ in~$U$.
Then, introduce a sequence of edges that form a directed path from $u_1$
to~$u_2$ through all the vertices of~$\overline U$,
say by taking the vertices in~$\overline U$ in ascending order
and using the first available incoming and outgoing edges of~$U$.  This makes
$X'$ path-connected. Third, add back in edges to
make $X^\prime$ 4-regular; e.g., take the vertices in~$X$ and their
missing incoming-edges and outgoing edges in sequential (ascending) order,
and join them up, sequentially. Note that the 4-regularity of~$X$ implies
that there are the same number of missing incoming edges as missing outgoing edges.

Now the universal traversal sequences (\ref{eq:AlgIndices}) based on
edge labels~$w$ should be universal
for both $X$ and~$X'$. Note, however, that each such traversal sequence has
to lead to the same vertex in~$X^\prime$ as in~$X$. This is because
the traversal never leaves~$U$, and this part of~$X^\prime$ is
the same as in~$X$.  Now we obtain a contradiction: Let $z\in \overline U$ have
an edge from~$u_1$ in~$X_1$.
It is not reachable by any traversal~(\ref{eq:AlgIndices}) in~$X$, but it must
reachable in~$X^\prime$ since $X^\prime$ is (subset-)connected.
\end{proof}

We get Theorem~\ref{thm:LeqSLinVL} as an immediate corollary.

\begin{proof}[Proof of Theorem \ref{thm:LeqSLinVL}]
Let $x$ be a vertex in an undirected graph~$G$.
Form the set $X_1$ as above; $X_1$~has multiple  vertices (one for each incident to~$x$ in~$G$)
that correspond to~$x$: pick one of them and denote it~$x'$. Then $H$ is the set
of vertices~$u$ in~$G$ that correspond to vertices~$u'$ in~$X_{m_1}$
such that there is a directed edge in~$X_{m_1}$ from $x'$ to~$u'$. This directed edge corresponds to
some path in~$X_1$ (as found by the $\Neighbor$
algorithm), and this path pulls back to a path
in the undirected graph~$G$. Putting these paths
together in a single string gives the desired~$Z$
encoding paths from~$x$ to all reachable vertices $u$ in~$G$.
\end{proof}

This completes the proof that $\L=\SL$ in~$\VL$.  The
logspace algorithm to determine whether there
is a path from $u$ to~$v$ in an undirected graph~$Y$,
acts by forming $X$ and~$X_m$, choosing (arbitrarily) a
pair of vertices $x,y$ in~$X$ that correspond (respectively) to
$u$ and~$v$, and checking whether one of the polynomially many
outgoing edges of~$x$ in~$X_{m-1}$ connects $x$ to~$y$.

\section{Cheeger and Sedrakyan inequalities in \texorpdfstring{$\VNC^1$}{VNC1}}\label{sec:SedrakyanCheeger}

The section proves that Sedrakyan's lemma (Lemma~\ref{lem:Sedrakyan})
and that part~(a) the Cheeger inequality (namely, the $2 \epsilon ~\ge~ \deltaGap$ part)
is provable in
$\VNC^1$, and hence in~$\VL$.  To formalize and prove these in
$\VNC^1$, vectors should be encoded as rational numbers with a
common denominator. 

The proof
of Sedrakyan's lemma based on Cauchy--Schwarz that was
given earlier used square roots, which means
that $\VNC^1$ cannot use this proof method using only rational
numbers.  Instead, we present a proof (based on one of the standard
proofs of the Cauchy--Schwarz inequality) that can be formalized
via reasoning only about rational numbers.

Our $\VNC^1$ proof of the first half of Cheeger inequality will
follow the first one of the four proofs exposited in
Chung~\cite{Chung:Cheeger}.  We have also used the
lecture notes of Sauerwald-Sub~\cite{SauerwaldSun:Cheeger}
to fill in
some of the missing details in~\cite{Chung:Cheeger}.
For simplicity, we will give the proof only for
the case of $d$-regular graphs. However, the proof
in~\cite{Chung:Cheeger} covers the more general case of part~(a)
of Cheeger's inequality, and is
also formalizable in~$\VNC^1$.

\subsection{Sedrakyan's lemma in \texorpdfstring{$\VNC^1$}{VNC1}}\label{sec:Sedrakyan}

\begin{theorem}\label{thm:SedrakyanVNC1}
$\VNC^1$~proves Sedrakyan's Lemma for
$n$-vectors $\vec u$ and $\vec v$ of rational numbers with
a common denominator provided each component $v_i$ of~$\vec v$ is positive.
\end{theorem}

\begin{proof}
The following argument is formalizable in~$\VNC^1$, where $i$ and~$j$ range over
labels for members of $\vec u$ and $\vec v$.

\begin{eqnarray*}
\lefteqn{\sum\nolimits_i \frac {u_i^2}{v_i} ~\ge~ \frac {\bigl(\sum\nolimits_i u_i\bigr)^2}{\sum_i v_i}}
\\ &\Leftrightarrow&
\bigl(\sum\nolimits_i u_i\bigr)^2 ~\le~ \sum\nolimits_i \frac {u_i^2}{v_i} \cdot \sum\nolimits_i v_i
\\ &\Leftrightarrow&
\sum\nolimits_{i<j} 2u_i u_j + \sum\nolimits_i u_i^2 ~\le~
     \sum\nolimits_i \sum\nolimits_j \frac{u_i^2}{v_i} v_j
\\ &\Leftrightarrow&
\sum\nolimits_{i<j} 2u_i u_j + \sum\nolimits_i u_i^2 ~\le~
   \sum\nolimits_i u_i^2 + \sum \nolimits_{i\not=j} \frac{u_i^2}{v_i} v_j
\\ &\Leftrightarrow&
\sum\nolimits_{i<j} 2u_i u_j ~\le~
   \sum\nolimits_{i<j}\Bigl( \frac{u_i^2}{v_i} v_j + \frac{u_j^2}{v_j} v_i \bigr)
\\ &\Leftrightarrow&
0 ~\le~ \sum\nolimits_{i<j} \bigl( u_i^2 \frac{v_j}{v_i} + u_j^2 \frac{v_i}{v_j} -2 u_i u_j \bigr)
\\ &\Leftrightarrow&
0 ~\le~ \sum\nolimits_{i<j} \bigl( u_i^2 v_j^2 + u_j^2 v_i^2 - 2 u_i u_j v_i v_j \bigr) \frac 1 {v_i v_j}
\\ &\Leftrightarrow&
0 ~\le~ \sum\nolimits_{i<j} ( u_i v_j - u_j v_i )^2 \frac 1 {v_i v_j}.
\end{eqnarray*}
The last inequality is obviously true as squares are nonnegative,
and since the $v_i$'s are positive.
\end{proof}

\subsection{Half of the Cheeger inequality in \texorpdfstring{$\VNC^1$}{VNC1}}\label{sec:Cheeger}

\begin{theorem}
$\VNC^1$ can prove part (a) of the Cheeger inequality in the form stated
in Theorem~\ref{thm:CheegerVNCone}.
\end{theorem}

Recall that part~(a) formalizes the inequality $2 \epsilon ~\ge~ \deltaGap$ constructively.

\begin{proof}
We prove part~(a) of Theorem~\ref{thm:CheegerVNCone}, arguing informally
using methods that can be formalized in~$\VNC^1$.
Assume that $U$ is a set of vertices in~$G$ that shows the edge expansion
as defined by~(\ref{eq:edgeExpansionDef}) in Definition~\ref{def:edgeExpansion}
is $< \alpha$.  Let $G$ have vertices~$[N]$ and edges~$E$; let $i \sim j$ mean that
$\{i,j\}$ is in the multiset of edges~$E$.
Let $\chi_U$ be the $N$-vector
such that its $i$-th component is
$(\chi_U)_i = 1$ if $i\in U$ and $(\chi_U)_i = 0$ for $i\notin U$.
Let $\vec v$ be the component of~$\chi_U$ that is perpendicular to~$\vecone$;
namely, $\vec v = \chi_U - (\langle \chi_U, \vecone\rangle / N^2)\vecone$.
More explicitly,
\[
v_i ~=~ \left\{ \begin{array}{ll}
1 - \frac{|U|}{N} \quad & \hbox{if $i \in U$} \\[0.3ex]
- \frac{|U|}{N} \quad & \hbox{if $i \notin U$.}
\end{array} \right.
\]
Our goal is to prove $\|M \vec v\|^2 \ge (1-2\alpha)^2 \|\vec v\|^2$ and
thereby prove part~(a).
Using the assumption of edge expansion, we get
\begin{equation}\label{eq:edgeExpAssume}
\sum_{i\sim j} ( v_i - v_j )^2 ~=~
    \sum_{i\sim j} ( (\chi_U)_i - (\chi_U)_j )^2 ~<~ \alpha \cdot d \cdot |U|.
\end{equation}
We have\footnote{We
work with $\|\vec v\|^2$
instead of $\|\vec v\|$ since $\VNC^1$ needs to talk about
square norms, not norms.}
\[
\|\vec v\|^2 ~=~ |U|\Bigl(1-\frac{|U|}{N}\Bigr)^2 + (N-|U|)\Bigl(\frac{|U|}{N}\Bigr)^2  
~=~
     |U|\Bigl( 1 - \frac{|U|}{N}\Bigr) .
\]
Consequently, $\|v\|^2 \ge |U|/2$ since, without loss of generality, $|U| \le N/2$.

\smallskip

The following holds
for an arbitrary vector~$\vec v$:
\begin{lemma}
Let $M$ be the (normalized) adjacency matrix for a $d$-regular undirected graph~$G$.
Then
\[
\sum_{i\sim j} ( v_i - v_j )^2 ~=~
   d \cdot ( \|\vec v\|^2 - \langle \vec v, M \vec v \rangle ).
\]
\end{lemma}
\begin{proof} (of the lemma).
The lemma is proved by:
\begin{align*}
\sum_{i\sim j} ( v_i - v_j )^2
   &=~
       \sum_{i\sim j} ( v_i^2 + v_j^2 - 2 v_i v_j )
   ~=~ \sum_i d \cdot v_i^2 - 2 \sum_{i\sim j} v_i v_j \\
&=~ d \cdot \|\vec v\|^2 - \sum_i \bigl(v_i \cdot \sum_{j : i\sim j} v_j \bigr) \\
&=~ d \cdot \|\vec v\|^2 - \sum_i \bigl(v_i \cdot d \cdot (M \vec v)_i \bigr) \\
&=~ d \cdot \|\vec v\|^2 - d \cdot \langle \vec v, M \vec v \rangle . \qedhere
\end{align*}
\end{proof}
The lemma and (\ref{eq:edgeExpAssume}) imply
$\langle \vec v, M \vec v \rangle \ge \|\vec v\|^2 - \alpha \cdot |U|$,
whence
\[
    \frac{\langle \vec v, M \vec v \rangle}{\|\vec v\|^2}
    ~\ge~ 1 - \frac{|U|}{\|\vec v\|^2} ~\ge~
    ~\ge~ 1 - 2 \alpha,
\]
since $|U|/2 \le \|\vec v\|^2$.  Since $\langle \vec v, M \vec v \rangle \vec v/ \|\vec v\|^2$
is the projection of $M\vec v$ onto~$\vec v$, the 
Cauchy--Schwarz inequality gives
\[
\|M\vec v\|^2 ~\ge~ \frac{ \| ( \langle \vec v, M \vec v \rangle \vec v) \|^2} { \|\vec v\|^4 } 
   ~=~ \frac{ \langle \vec v, M \vec v \rangle^2 } {\|\vec v\|^2} .
\]
Hence $\|M\vec v\|^2 \ge (1-2\alpha)^2 \|\vec v\|^2$. This establishes part~(a)
of Theorem~\ref{thm:CheegerVNCone} since the second-order~$V$ can be the set encoding
the vector~$\vec v$.
\end{proof}

\paragraph*{Acknowledgement.} We thank the two 
anonymous referees for useful comments and corrections.

\printbibliography


\end{document}